\documentclass{IEEEtran}

\usepackage[cmex10]{amsmath}
\usepackage{mathrsfs}
\usepackage{cite}
\usepackage{hyperref}
\usepackage[caption=false,font=footnotesize,
	subrefformat=parens]{subfig}
\usepackage{fp}

\usepackage{amsthm}
\usepackage{amssymb}
\usepackage{tikz}
\usetikzlibrary{calc}

\newtheorem{theorem}{Theorem}
\newtheorem{corollary}{Corollary}
\newtheorem{proposition}{Proposition}
\newtheorem{lemma}{Lemma}

\theoremstyle{definition}
\newtheorem{definition}{Definition}

\renewcommand{\qed}{\hfill\IEEEQED}

\newcommand{\Conline}{\Theta}
\newcommand{\T}{\mathscr{T}}

\begin{document}
\allowdisplaybreaks

\title{Universally Near Optimal Online Power Control for Energy Harvesting Nodes}

\author{Dor Shaviv and Ayfer \"Ozg\"ur
\thanks{The authors are with the Department of Electrical Engineering, Stanford University.}}

\maketitle

\begin{abstract}
We consider online power control for an energy harvesting system with random i.i.d. energy arrivals and a finite size battery. We propose a simple online power control policy for this channel that requires minimal information regarding the distribution of the energy arrivals and prove that it is universally near-optimal for all parameter values. In particular, the policy depends on the distribution of the energy arrival process only through its mean and it achieves the optimal long-term average throughput of the channel within both constant additive and multiplicative gaps. Existing heuristics for online power control fail to achieve such universal performance. This result also allows us to approximate the long-term average throughput of the system with a simple formula, which sheds some light on the qualitative behavior of the throughput, namely how it depends on the distribution of the energy arrivals and the size of the battery.

\end{abstract}

\IEEEpeerreviewmaketitle

\section{Introduction}

Recent advances in energy harvesting technologies enable
wireless devices to harvest the energy they need for communication from the natural
resources in their environment. This development opens
the exciting possibility to build wireless networks that are
self-powered, self-sustainable and which have lifetimes limited
by their hardware and not the size of their batteries.

Communication with such wireless devices requires the design of good power control policies that can maximize throughput under random energy availability. In particular, available energy should not be consumed
too fast, or transmission can be interrupted in the future due to an
energy outage; on the other hand, if the energy consumption
is too slow, it can result in the wasting of the harvested
energy and missed recharging opportunities in the future due
to an overflow in the battery capacity.

The problem of power control for energy harvesting communication has received significant interest in the recent literature~\cite{YangUlukus2012,TutuncuogluYener2012,Ozeletal2011,
DP0,DP1,DP2,DP3,ho2012optimal,MaoCheungWong2012,Wang15,Kazerouni2015,Mitran,
online_infiniteB1,online_infiniteB2,online_infiniteB3,online_infiniteB4,
info2013,amirnavaei2015online,xu2014throughput}.
In the offline case, when future energy arrivals are known ahead of time, the problem has an explicit solution~\cite{YangUlukus2012,Ozeletal2011,TutuncuogluYener2012}. The optimal policy keeps energy consumption as constant as possible over time while ensuring no energy wasting due to an overflow in the battery capacity. The more interesting case is the online scenario where future energy arrivals are random and unknown. In this case, the problem can be modeled as a Markov Decision Process (MDP) and solved numerically using dynamic programming~\cite{DP0,DP1,DP2,DP3,ho2012optimal,MaoCheungWong2012,Wang15}. However, this numerical approach has several shortcomings. First, although there exist numerical methods to find a solution which is arbitrarily close to optimal, such
as value iteration and policy iteration, these methods rely
on quantization of the state space and the action space.
Specifically, the computational load of each iteration grows
as the cube of the number of quantized states and/or actions and may not be suitable for sensor nodes with limited computational capabilities.
Second, the solution depends on the exact model for the statistical distribution of the
energy arrival process; this may be hard to obtain in practical
scenarios and may require recomputing the dynamic programming solution periodically to track changes in the harvesting process. Finally, the numerical solution does not provide much insight on the structure of the optimal online power control policy and the
qualitative behavior of the resultant throughput, namely how it
varies with the parameters of the problem. This kind of insight
can be critical for design considerations, such as choosing the
size of the battery to employ at the transmitter.  

Due to these limitations of the numerical approach, there has been significant effort in the recent literature to develop simple heuristic online policies. These policies come either with no guarantees or only asymptotic guarantees on optimality~\cite{Mitran,
online_infiniteB1,online_infiniteB2,online_infiniteB3,
online_infiniteB4,info2013,amirnavaei2015online}. For example, two natural heuristic policies, the variations of which have been widely studied in the literature, are the greedy policy and what we will refer to as the constant policy. In the greedy policy, the transmitter instantaneously uses all the harvested energy in each time slot. This greedy approach ensures that energy is never wasted due to an overflow in the battery capacity and it naturally becomes optimal in the limit when the harvested energy (or SNR) goes to zero, since in the low-SNR limit the capacity becomes proportional to the total energy transferred to the receiver.
The constant policy on the other hand aims to keep power allocation as constant as possible over time, for example by always allocating energy equal to the mean $\mu$ of the energy harvesting process as long as there is sufficient energy in the battery.
It is easy to see that this policy becomes optimal in the limit when the battery size goes to infinity, since an infinite battery can allow to average out the randomness in the harvesting process and therefore allocate energy equal to $\mu$ in almost all time-slots. However, such asymptotic results do not provide any insights on the gap to optimality for given finite parameter values. In particular, given a certain distribution of the energy harvesting process and a given finite size of the battery, these two policies can be arbitrarily away from optimality and it is not even clear which one of these two structurally very different policies would be a better choice for the given problem.

The goal of this paper is to address this problem. Instead of seeking policies that become asymptotically optimal in one limit or the other, we look for policies that are provably close to optimal across all parameter regimes and any distribution of the energy arrivals. In particular, we seek policies that achieve the optimal long-term average throughput of the system simultaneously within a constant multiplicative factor and a constant additive gap for all parameter values and distributions of the energy arrivals. This more stringent requirement ensures that these policies are universally near-optimal in the sense that their gap to optimality remains bounded across all parameter ranges. It would be moreover desirable for these policies to have minimal dependence on the distribution of the energy arrivals, for example depend only on the mean energy arrival rate. This would enable one to apply them to  any given problem with arbitrary parameter values, without even knowing the exact distribution of the energy arrivals, while she/he would be assured to achieve a performance that is very close to the one achieved by an optimal policy specifically optimized for the given problem, in particular the exact distribution of the energy arrivals. While the two policies discussed above can be applied with minimal knowledge of the energy arrival distribution (the greedy policy does not require any information about the distribution while the constant policy is based only on the knowledge of the mean), it is easy to show that neither of them is universally near-optimal; there are parameter regimes and distributions for which their gap to optimality grows unboundedly.  

The main result of this paper is to propose a simple novel online power control policy which depends on the harvesting process only through its mean: at each time-slot, the policy uses a constant fraction of the energy available in the battery where the fraction is chosen as the ratio of the mean of the energy arrival distribution and the battery size. This policy is structurally very different from the two policies discussed above and may appear a priori counter-intuitive. We show that it is naturally motivated by the case where the energy arrival process is i.i.d. Bernoulli, in which case the optimal online power control policy can be explicitly characterized. We then establish the near-optimality of this policy for any i.i.d. harvesting process and any size of the battery. In particular, we show that this policy  achieves the optimal long-term average throughput of the system simultaneously within a constant multiplicative factor and a constant additive gap for all parameter values. This implies that this policy can be applied under any i.i.d. harvesting process, without even knowing the statistical distribution of the energy arrivals. The multiplicative and additive approximations guarantee that it will perform close to the best strategy optimized for the exact distribution of the energy arrivals across all parameter values (both in the high- and the low-SNR regimes). The main ingredient of our proof is to show that for the proposed policy Bernoulli is the worst case distribution for the energy arrivals. Therefore, the performance of the scheme under a Bernoulli distribution provides a lower bound on its performance under any i.i.d. process. In this sense,  our policy can be thought of as building on the insights from the worst-case scenario, hence performs well in the worst-case sense, while previous heuristics can be thought of as building on the insights from the best case scenario, i.e. when energy arrivals are constant equal to $\mu$.

This result also leads to  a simple approximation of the optimal long-term average throughput of an AWGN energy harvesting channel. In particular, we show that within a constant gap, the average throughput is given by
\[
\Theta\approx\frac{1}{2}\log\left(1+\mathbb{E}[\min\{E_t,\bar{B}\}]\right),
\]
where $E_t$ denotes the energy arrival process and $\bar{B}$ is the battery capacity. This shows that a battery large enough to capture the maximal energy arrival over a single time-slot is sufficient to approximately achieve the AWGN capacity.

\section{System Model}
\label{sec:system_model}

We begin by introducing the notation used throughout the paper.
Let $\mathbb{E}[\,\cdot\,]$ denote expectation.
For $m\leq n$, denote $X_m^n=(X_m,X_{m+1},\ldots,X_{n-1},X_n)$ and $X^n=X_1^n$.
All logarithms are to base~2, and $\ln$ will denote log to base $e$.

We consider a point-to-point single user channel with additive white Gaussian noise (AWGN). We assume a quasi-static fading channel, in which the channel coefficient remains constant throughout the entire transmission time.
The system operation is slotted, i.e. time is discrete ($t=1,2,\ldots$).
At time $t$, the received signal is $y_t=\sqrt{\gamma}x_t+z_t$, where $x_t$ is the transmitted signal, $\gamma$ is the channel coefficient or SNR, and $z_t$ is unit-variance zero-mean white Gaussian noise.

The transmitter is equipped with a battery of finite capacity $\bar{B}$.
Let $E_t$ be energy harvested at time $t$, which is assumed to be a non-negative i.i.d. process with $\mathbb{E}[E_t]>0$.
We assume the energy arrival process is known causally at the transmitter.
A power control \emph{policy} for an energy harvesting system is a sequence of mappings from energy arrivals to a non-negative number, which will denote a level of instantaneous power.
In this work, we will focus on \emph{online policies}.
An {online policy} $\mathbf{g}=\{g_t\}_{t=1}^{\infty}$ is a sequence of mappings 
\begin{equation}\label{eq:online}
g_t:\mathcal{E}^t\to\mathbb{R}_+\qquad
,t=1,2,\ldots.
\end{equation}
By allocating power $g_t$ at time $t$, the resultant instantaneous rate is $r_t=\frac{1}{2}\log(1+\gamma g_t)$.

Let $b_t$ be the amount of energy available in the battery at the beginning of time slot $t$. 
An \emph{admissible policy} $\mathbf{g}$ is such that satisfies the following constraints for every possible harvesting sequence $\{E_t\}_{t=1}^{\infty}$:
\begin{align}
&0\leq g_t\leq b_t&&,t=1,2,\ldots,\label{eq:power}\\
&b_t=\min\{b_{t-1}-g_{t-1}+E_t,\bar{B}\}&&,t=2,3,\ldots,\label{eq:battery}
\end{align}
where we assume $b_1=\bar{B}$ without loss of generality.

For a given policy $\mathbf{g}$, we define the $n$-horizon expected throughput to be:
\begin{equation}
\T_n(\mathbf{g})=\frac{1}{n}\mathbb{E}\left[
\sum_{t=1}^{n}\frac{1}{2}\log(1+\gamma g_t(E^t))\right],
\label{eq:def_throughput}
\end{equation}
where the expectation is over the energy arrivals $E_1,\ldots,E_n$.
Finally, our goal is to characterize the optimal online power control policy and the resultant long-term average throughput:
\begin{align}
\Conline
&=\sup_{\mathbf{g}\text{ admissible}}
\liminf_{n\to\infty}\T_n(\mathbf{g}).
\label{eq:online_opt}
\end{align}

\section{Preliminary Discussion}
\label{sec:preliminary_discussion}

While the optimal offline power control policy has been explicitly characterized in~\cite{YangUlukus2012,Ozeletal2011,TutuncuogluYener2012}, in which the energy arrival sequence $E^n$ is assumed to be known ahead of time, there is limited understanding regarding the structure of the optimal online power control policy and the resultant long-term average throughput. 

It is easily observed that this is an MDP, with the state being the battery level $b_t$, the action $g_t$ allowed to take values in the interval $[0,b_t]$, and the disturbance $E_{t+1}$. The state dynamics are given by
\begin{equation}
b_{t+1}=f(b_t,g_t,E_{t+1})\triangleq\min\{b_t-g_t+E_{t+1},\bar{B}\},
\label{eq:DP_state_dynamics}
\end{equation}
and the stage reward is $r_t=\frac{1}{2}\log(1+\gamma g_t)$.
It then follows by a well-known result in MDPs~\cite[Theorem~4.4.2]{Puterman2005} that the optimal policy is \emph{Markovian}, i.e. it depends only on the current state: $g_t^\star(E^t)=g_t^\star(b_t)$.
If the policy depends only on the current state and it is time-invariant, i.e. $g_t^\star(E^t)=g^\star(b_t)$, we say it is \emph{stationary}.
The optimal policy can be found by means of dynamic programming, which involves solving the Bellman equation:
\begin{proposition}[{Bellman Equation~\cite[Theorem~6.1]{arapostathis1993discrete}}]
If there exists a scalar $\lambda\in\mathbb{R}_+$ and a bounded function $h:[0,\bar{B}]\to\mathbb{R}_+$ that satisfy
\begin{align}
&\lambda+h(b)=\sup_{0\leq g\leq b}\big\{
\tfrac{1}{2}\log(1+\gamma g)\nonumber\\*
&\hspace{10em}
+\mathbb{E}
[h(\min\{b-g+E_t,\bar{B}\})]
\big\}
\label{eq:Bellman}
\end{align}
for all $0\leq b\leq\bar{B}$,
 then the optimal throughput is $\Conline=\lambda$.
Furthermore, if $g^\star(b)$ attains the supremum in~\eqref{eq:Bellman} then the optimal policy is stationary and is given by $g_t^\star(E^t)=g^\star(b_t(E^t))$.
\label{prop:Bellman}
\end{proposition}

The functional equation~\eqref{eq:Bellman} is hard to solve explicitly, and requires an exact model for the statistical distribution of the energy arrivals $E_t$, which may be hard to obtain in practical scenarios. The equation can be solved numerically using value iteration~\cite{Bertsekas2001vol2}, but this can be computationally demanding and the numerical solution cannot provide insight as to the structure of the optimal policy and the qualitative behavior of the optimal throughput, namely how it varies with the parameters of the problem. 

In the sequel, we propose an explicit online power control policy and show that it is within a constant gap to optimality for all i.i.d. harvesting processes. This policy depends on the harvesting process only through its mean, and does not depend on the channel gain whereas the optimal solution may depend on $\gamma$. It also leads to a simple and insightful approximation of the achievable throughput. We first discuss a special case in which the optimal online solution can be explicitly found. This inspires the approximately optimal power control policy for general i.i.d. energy harvesting processes.

\section{Bernoulli Energy Arrivals}
\label{sec:bernoulli_energy_arrivals}

Assume the energy arrivals $E_t$ are i.i.d. Bernoulli random variables (RVs):
\begin{equation}\label{eq:Bernoulli_Et}
	E_t=
	\begin{cases}
		\bar{B}&\text{w.p. }p\\
		0&\text{w.p. }1-p,
	\end{cases}
\end{equation}
i.e. at each time $t$ either the battery is fully charged to $\bar{B}$ with probability $p$ or no energy is harvested at all with probability $1-p$.
This simple case was studied extensively in~\cite{ShavivOzgur2015,Kazerouni2015}, and was shown there to be solved exactly. Specifically, we have the following Theorem, which we prove in Appendix~\ref{sec:kkt_solution}.
\begin{theorem}\label{thm:Bernoulli}
Let the energy harvesting process be defined by~\eqref{eq:Bernoulli_Et}.
Let $j_t(E^t)$ be the time of the last energy arrival, i.e.
\[
j_t(E^t) = \{\sup\ \tau\leq t:\ E_{\tau}=\bar{B}\}.
\]
Then the optimal policy is given by
\begin{equation}
g_t^\star(E^t)=\begin{cases}
\frac{1}{\gamma}\left(\frac{\tilde{N}+\gamma\bar{B}}{1-(1-p)^{\tilde{N}}}p(1-p)^{t-j_t}-1\right)
	&,t-j_t<\tilde{N}\\
0	&,t-j_t\geq\tilde{N}
\end{cases}
\label{eq:Bernoulli_gt_solution}
\end{equation}
where $\tilde{N}$ is the smallest positive integer satisfying 
\[
1 > (1-p)^{\tilde{N}}[1+p(\gamma\bar{B}+\tilde{N})].
\]
\end{theorem}

It can be seen that this is a stationary policy, i.e. $g_t^\star$ can be written as a time-invariant function of $b_t$.
Roughly speaking, the energy is allocated only to the first $\tilde{N}$ time slots after each battery recharge and decays in an approximately exponential manner.

For the purpose of extending this policy to general i.i.d. processes in the next section, it is useful to simplify it to the following form by preserving its exponentially decaying structure: 
\begin{equation}\label{Bernoulliform1}
g_t(E^t)=\bar{B}p(1-p)^{t-j_t},
\end{equation}
where $j_t$ is the time of the last energy arrival, as defined above. With this simplified policy, the amount of energy we allocate to each time slot decreases exactly exponentially with the time elapsed since the last battery recharge (or equivalently energy arrival). 
Note that this is clearly an admissible strategy since 
$$
\sum_{k=j_t}^\infty\bar{B}p(1-p)^{k-j_t}=\bar{B},
$$
i.e. the total energy we allocate until the next battery recharge can never exceed $\bar{B}$, the amount of energy initially available in the battery.
Another way to view this strategy is that we always use $p$ fraction of the remaining energy in the battery, i.e. 
\begin{equation}\label{Bernoulliform2}
g_t(b_t)=pb_t,
\end{equation} where $b_t$ is the available energy in the battery given by 
$b_t= (1 - p)^{t-j_t}\bar{B}$.
Hence, it is a stationary policy.

This simplified policy can be intuitively motivated as follows: for the Bernoulli
arrival process $E_t$, the inter-arrival time is a geometric random
variable with parameter $p$. Because the geometric
random variable is memoryless and has mean $1/p$,
at each time step the expected time to the next
energy arrival is $1/p$. Since $\log(\cdot)$ is a concave
function, uniform allocation of energy maximizes throughput, i.e. if the current energy level
in the battery is $b_t$ and we knew that the next battery recharge  would be in exactly $m$ channel uses, allocating $b_t/m$ energy to each of the next $m$ channel uses would maximize throughput. For the online case of interest here, we can instead use the expected time to the next energy arrival: since at each time
step, the expected time to the next energy arrival is $1/p$, we
always allocate a fraction $p$ of the currently available energy in the battery.   
Fig.~\ref{fig:Bernoullipowercontrol} illustrates this power control policy.

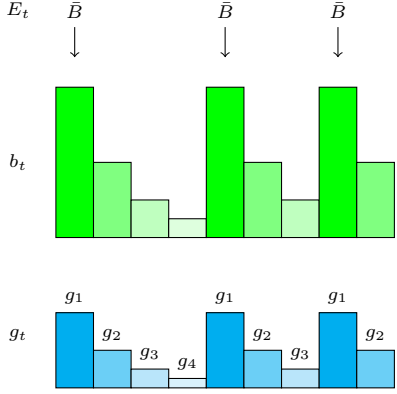
\begin{figure}
\centering
\begin{tikzpicture}
	\tikzstyle{every node}=[font=\scriptsize];
	
	\foreach \x/\g/\i in {0/1/1, 0.5/0.5/2, 1/0.25/3, 1.5/0.125/4, 2/1/1, 2.5/0.5/2, 3/0.25/3, 3.5/1/1, 4/0.5/2}
	{
		\FPeval{\mycolor}{round(100*\g,0)};
		\filldraw[fill=cyan!\mycolor!white,draw=black] (\x,0) -- (\x+.5,0) -- (\x+.5,\g) -- (\x,\g) -- cycle;
		\node[above] at (\x+.25,\g) {$g_{\i}$};
	}
	
	\node at (-.5,.7) {$g_t$};	
	
	\def \h {2};
	
	\foreach \x/\b/\e/\i/\mygreen/\myyellow in {0/2/\bar{B}/1/100/100, 0.5/1/0/2/0/0, 1/0.5/0/3/0/100, 1.5/0.25/0/4/30/100, 2/2/\bar{B}/5/100/100, 2.5/1/0/6/0/0, 3/0.5/0/7/40/100, 3.5/2/\bar{B}/8/100/100, 4/1/0/9/0/40}
	{
		\FPeval{\mycolor}{round(50*\b,0)};
		\filldraw[fill=green!\mycolor!white,draw=black] (\x,\h) -- (\x+.5,\h) -- (\x+.5,\h+\b) -- (\x,\h+\b) -- cycle;
	}
	
	\foreach \x in {0,2,3.5}
	{
		\draw[->] (\x+.25,\h+2.8) node[above] {$\bar{B}$} -- (\x+.25,\h+2.4);
	}
	
	\node[above] at (-.5,\h+2.8) {$E_t$};
	\node at (-.5,\h+1) {$b_t$};
\end{tikzpicture}
\caption{The approximately optimal online power control policy for Bernoulli energy arrivals.}
\label{fig:Bernoullipowercontrol}
\end{figure}

While this simplified policy is clearly suboptimal, it was shown in \cite{DongFarniaOzgur2015} to be at most within a gap of $0.97$ to optimality for all values of $\bar{B}$ and $p$. Here we improve upon this bound. Before we state this result, we present the following proposition which provides a simple upper bound on the achievable throughput for any i.i.d. harvesting process. The proof follows from Jensen's inequality, and can be found in e.g.~\cite{Mitran,online_infiniteB4}.
For completeness, we provide the proof in Appendix~\ref{sec:upper_bound}.

\begin{proposition}\label{prop:upperbound}
The optimal throughput under any i.i.d. harvesting process $E_t$ is bounded by
\[
\Conline
\leq\frac{1}{2}\log(1+\gamma\mu),
\]
where $\mu\triangleq \mathbb{E}[\min\{E_t,\bar{B}\}]$.
\end{proposition}

We next lower bound the performance of our simplified policy in terms of this upper bound.
\begin{proposition}\label{prop:Bernoulli_gap}
Let $E_t$ be given by~\eqref{eq:Bernoulli_Et} and consider the policy $\mathbf{g}$ given by~\eqref{Bernoulliform1} (or equivalently ~\eqref{Bernoulliform2}).
Then the gap to optimality is bounded as follows:
\begin{equation}
\liminf_{n\to\infty}\T_n(\mathbf{g})\geq \frac{1}{2}\log(1+\gamma\mu)-0.72.
\label{eq:Bernoulli_lower_bound_final}
\end{equation}
\end{proposition}
See Section~\ref{subsec:Bernoulli} for the proof.

The two propositions above state that the simplified policy is always within 0.72 bits/channel use of optimality. Numerical evaluations show, however, that the real gap  to optimality is much smaller than the one given in Proposition~\ref{prop:Bernoulli_gap}. In fact, Fig.~\ref{fig:Bernoulli_throughput_plots} shows that the throughput obtained by our simplified suboptimal policy is almost indistinguishable from the optimal throughput.

\begin{figure}
\centering
\includegraphics[width=21pc]{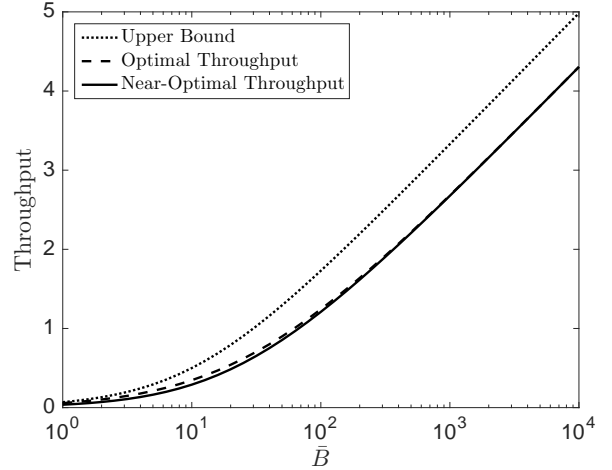}
\caption{Optimal throughput and near-optimal throughput using the simplified Bernoulli policy, for $p=0.1$ and $\gamma=1$.
These are plotted along with the upper bound suggested in Proposition~\ref{prop:upperbound}.}
\label{fig:Bernoulli_throughput_plots}
\end{figure}

An additive bound is especially useful when $\bar{B}$ is very large; as seen in Fig. \ref{fig:Bernoulli_throughput_plots}, the gap remains constant while the throughput grows unboundedly.
For small $\bar{B}$, however, the additive gap becomes useless as the optimal throughput itself falls below $0.72$.
Yet, it should be clear from the figure that as the throughput becomes small, the numerically evaluated additive gap also tends to zero.
We capture this fact in the following proposition where we provide a bound on the \emph{ratio} between the throughput achieved by the simplified Bernoulli policy and the  upper bound on the throughput.
\begin{proposition}
\label{prop:Bernoulli_mult}
For i.i.d. energy arrivals given by \eqref{eq:Bernoulli_Et} and the policy given by \eqref{Bernoulliform1}, the ratio to optimality is bounded below as follows:
\[
\liminf_{n\to\infty}\T_n(\mathbf{g})\geq \frac{1}{2}\cdot\frac{1}{2}\log(1+\gamma\mu) .
\]
\end{proposition}
The proof is provided in Section~\ref{subsec:Bernoulli_mult}.

The two propositions together imply that the simplified Bernoulli policy is good for all values of $\bar{B}$, i.e. across all SNR regimes.

\section{Approximately Optimal Policy for General i.i.d. Energy Arrival Processes}
\label{sec:suboptimal_policies}

We now assume that $E_t$ is an i.i.d. process with an arbitrary distribution. As discussed in Section~\ref{sec:preliminary_discussion}, finding the optimal solution for this general case is a hard problem.
In this section, we present a natural extension of the approximately optimal policy~\eqref{Bernoulliform2} in the Bernoulli case and show that it is approximately optimal for all i.i.d. harvesting processes. The policy reduces to~\eqref{Bernoulliform2} when the harvesting process is Bernoulli.

Before presenting the policy, observe that without loss of generality we can replace the random process $E_t$ with $\tilde{E}_t\triangleq\min\{E_t,\bar{B}\}$ without changing the system because of the store-and-use model we assume in \eqref{eq:power} and \eqref{eq:battery}. This is due to the fact that whenever an energy arrival $E_t$ is larger than $\bar{B}$, it will be clipped to at most $\bar{B}$. Denote $\mu=\mathbb{E}[\tilde{E}_t]$.

\paragraph*{The Fixed Fraction Policy}

Let $q\triangleq \mu/\bar{B}$. Note that $\mu\in[0,\bar{B}]$ so $q\in[0,1]$.
We will use $q$ here instead of the parameter $p$ in the Bernoulli case. Notice that in that case, we also have $\mathbb{E}[E_t]=p\bar{B}$, hence this is a natural definition.
The Fixed Fraction Policy is defined as follows:
\[
g_t=q b_t.\qquad t=1,2,\ldots
\]
Inspired by \eqref{Bernoulliform2}, at each time slot, this policy allocates a fraction $q$ of the currently available energy in the battery.
Clearly this is an admissible policy, since $q\leq 1$.

\subsection{Main Result}



The main result of this paper is that the Fixed Fraction Policy achieves the upper bound in Proposition~\ref{prop:upperbound} within a constant additive gap and a constant multiplicative factor for any i.i.d. process.
We prove this result by showing that under this policy, the Bernoulli harvesting process yields the \emph{worst} performance compared to all other i.i.d. processes with the same mean~$\mu$.
This implies that the lower bounds obtained for the throughput achieved under  Bernoulli energy arrivals apply also to any i.i.d. harvesting process with the same mean, giving the following theorem.

\begin{theorem}\label{thm:onlinePC}
Let $E_t$ be an i.i.d. non-negative process with $\mu=\mathbb{E}[\min\{E_t,\bar{B}\}]$, and let $\mathbf{g}$ be the Fixed Fraction Policy.
Then, the throughput achieved by $\mathbf{g}$ is bounded by
\[
\liminf_{n\to\infty}\T_n(\mathbf{g})\geq\frac{1}{2}\log(1+\gamma\mu)-0.72,
\]
and
\[
\liminf_{n\to\infty}\T_n(\mathbf{g})\geq\frac{1}{2}\cdot\frac{1}{2}\log(1+\gamma\mu).
\] 
\end{theorem}

The proof of this theorem is given in Section~\ref{subsec:general_energy_arrivals}. The following approximation for the optimal throughput is an immediate corollary of the above theorem and proposition.
\begin{corollary}\label{corollary}
The optimal throughput under any i.i.d. energy harvesting process $E_t$ is bounded by
\[
\frac{1}{2}\log(1+\gamma\mu)-0.72\leq\Conline
\leq\frac{1}{2}\log(1+\gamma\mu),
\]
and
\[
\frac{1}{2}\leq\frac{\Conline}{\frac{1}{2}\log(1+\gamma\mu)}\leq 1,
\]
where $\mu\triangleq \mathbb{E}[\min\{E_t,\bar{B}\}]$.
\end{corollary}

This corollary gives a simple approximation of how the optimal throughput depends on the energy harvesting process $E_t$ and the battery size $\bar{B}$. This characterization identifies two fundamentally different operating regimes for this channel where the dependence of the average throughput on $E_t$ and $\bar{B}$ is qualitatively different. Assume that $E_t$ takes values in the interval $[0, \bar{E}]$. When $\bar{B}\geq \bar{E}$, we have
\begin{equation}\label{eq:mainres1}
\Conline\approx \frac{1}{2}\log\left(1+\gamma\mathbb{E}[E_t]\right)\qquad\text{bits/s/Hz},
\end{equation}
and the throughput is approximately equal to the capacity of an AWGN channel with an average power constraint equal to the average energy harvesting rate. This is surprising given that the transmitter is limited by the additional constraints~\eqref{eq:power} and~\eqref{eq:battery},  and at finite $\bar{B}$ this can lead to part of the harvested energy being wasted due to an overflow in the battery capacity.  Note that in this large battery regime, the throughput depends only on the mean of the energy harvesting process -- two energy harvesting profiles are equivalent as long as they provide the same energy on the average -- and is also independent of the battery size $\bar{B}$. In particular, choosing $\bar{B}\approx \bar{E}$ is almost sufficient to achieve the throughput at infinite battery size.\looseness=-1

\begin{figure}[t]
\centering
\subfloat[$\bar{B}\geq\bar{E}$]{
\begin{tikzpicture}
\tikzstyle{every node}=[font=\scriptsize];

\draw[->] (0,0) -- (3,0);
\draw[->] (0,0) -- (0,2);
\node[below] at (0,0) {$0$};
\node[above] at (0,2) {$f_{\tilde{E}}(x)$};
\node[right] at (3,0) {$x$};

\draw[red,thick] (0,0.75) to[out=0,in=180] (0.75,1.25);
\draw[red,thick] (0.75,1.25) to[out=0,in=135] (1.75,0);
\node[below] at (1.75,-0.1) {$\bar{E}$};
\draw (1.75,0) -- (1.75,-0.1);

\draw[dashed,thick] (2.5,0) -- (2.5,2);
\node[below] at (2.5,-0.1) {$\bar{B}$};
\draw (2.5,0) -- (2.5,-0.1);

\end{tikzpicture}
\label{subfig:large_battery_regime}
}
\ 
\subfloat[$\bar{B}\leq\bar{E}$]{
\begin{tikzpicture}

\tikzstyle{every node}=[font=\scriptsize];

\draw[->] (0,0) -- (3,0);
\draw[->] (0,0) -- (0,2);
\node[below] at (0,0) {$0$};
\node[above] at (0,2) {$f_{\tilde{E}}(x)$};
\node[right] at (3,0) {$x$};

\draw[red,thick] (0,0.75) to[out=0,in=180] (0.75,1.25);
\draw[red,thick] (0.75,1.25) to[out=0,in=140] (1.1,1.12);

\node[below] at (1.1,-0.1) {$\bar{B}$};
\draw (1.1,-0.1) -- (1.1,0);
\draw[->,red,thick,>=latex] (1.1,0) -- (1.1,1.5);

\end{tikzpicture}
\label{subfig:small_battery_regime}
}
\caption{$\tilde{E}_t=\min\{E_t,\bar{B}\}$ in the two battery regimes.}
\label{fig:two_battery_regimes}
\end{figure}
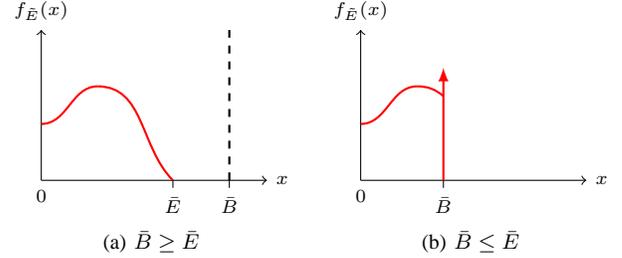

When $\bar{B}\leq \bar{E}$, note that one can equivalently consider the distribution of $E_t$ to be that in Fig.~\ref{fig:two_battery_regimes}-\subref{subfig:small_battery_regime}: since every energy arrival with value $E_t\geq\bar{B}$ fully recharges the battery, this creates a point mass at $\bar{B}$ with value $\Pr(E_t\geq\bar{B})$. In this case, Corollary~\ref{corollary} reveals that the throughput is approximately determined by the mean of this modified distribution. This can be interpreted as the small battery regime of the channel. In particular, in this regime the achievable throughput depends both on the shape of the distribution of $E_t$ and the value of~$\bar{B}$.

\section{Numerical Results}
\label{sec:conclusion}

We compare the Fixed Fraction Policy to two additional policies that have been studied in the literature: The greedy policy~\cite{Wang15}, i.e. $g_t=b_t$, and the ``constant'' policy \cite{Ozeletal2011},
in which $g_t=\mu\cdot\mathbf{1}\{b_t\geq\mu\}$, where $\mu=\mathbb{E}[\min\{E_t,\bar{B}\}]$. The latter attempts to transmit at a constant power $\mu$, and if there is not enough energy in the battery, it simply waits until the battery is recharged again to a level at least $\mu$.\footnote{Reference~\cite{online_infiniteB4} considers a refined version of this policy where the allocated energy is either $\mu+\delta$ or $\mu-\delta$ where $\delta=\beta\sigma^2\frac{\log\bar{B}}{\bar{B}}$ for some constant $\beta\geq2$ and $\sigma^2$ is the variance of the energy arrival distribution.
They show this policy is asymptotically optimal if $\bar{B}\to\infty$ and $\sigma^2$ is finite, in which case $\delta\to0$ and the strategy approaches the constant strategy. However, this policy is not applicable for all finite values of $\bar{B}$. For example, consider i.i.d. energy arrivals uniformly distributed on $[0,\bar{B}]$. Then $\sigma^2=\frac{1}{12}\bar{B}^2$, which yields $\delta=\frac{\beta}{12}\bar{B}\log\bar{B}$. Observe that for $\bar{B}<1$ we get $\delta<0$, and for $\bar{B}>e^{6/\beta}$ we get $\delta>\mu$, where $\mu=\bar{B}/2$. For $\beta=2$, which is the minimal value of $\beta$ according to~\cite{online_infiniteB4}, the policy is only applicable for $1\leq\bar{B}\leq 20.1$.
Therefore we do not include this policy in our numerical evaluations.
}

All policies are compared to the optimal throughput $\Conline$ obtained by dynamic programming via value iteration (See Section~\ref{sec:preliminary_discussion}).
Figures \ref{fig:sim_bernoulli0.1}, \ref{fig:sim_bernoulli0.5} and \ref{fig:sim_bernoulli0.9} depict the 
performance of these policies
when the energy distribution is Bernoulli with $p=0.1$, $p=0.5$ and $p=0.9$ respectively.  
The plots are, from left to right: the achievable throughput, the additive gap to optimality, and the ratio between the throughput and the optimal one.
The figures show the absolute attainable throughput using these policies, as well as the gap to optimality as measured by $\Conline-\liminf_{n\to\infty}\T_n(\mathbf{g})$ and the ratio $\frac{\liminf_{n\to\infty}\T_n(\mathbf{g})}{\Conline}$ between the suboptimal policy and the optimal one, where $\mathbf{g}$ is any of the policies mentioned above.
Note that for the Fixed Fraction Policy, the gap to optimality remains small in all cases, while the gap to optimality grows unboundedly for the greedy and the constant policy as $\bar{B}$ increases. In particular, note that although the greedy policy is (not surprisingly) the best of all strategies at small $\bar{B}$ values, which correspond to low-SNR, its performance is not universally good across all SNR regimes. This can be also observed by noting that for Bernoulli arrivals the throughput achieved by the greedy policy is given by $p\frac{1}{2}\log(1+\bar{B})$ while the optimal throughput is $\frac{1}{2}\log(1+p\bar{B})$ within $0.72$ bits/channel use as shown in Proposition~\ref{prop:Bernoulli_gap}. Obviously, when $p$ is small and $\bar{B}$ is large the gap between the two expressions can be arbitrarily large. Also, note that as $p$ increases, the gap to optimality decreases for all strategies. This is not surprising as when $p\rightarrow 1$, the Bernoulli distribution approaches a constant equal to $\mu=\bar{B}$. In this trivial case, all three policies reduce to the optimal policy that always allocates energy equal to $\mu$. However note that even with $p=0.9$, the fixed fraction policy is still able to provide gain over the other strategies.  

\def \smallfigwidth {12.5pc}
\begin{figure*}
\centering
\includegraphics[width=\smallfigwidth]{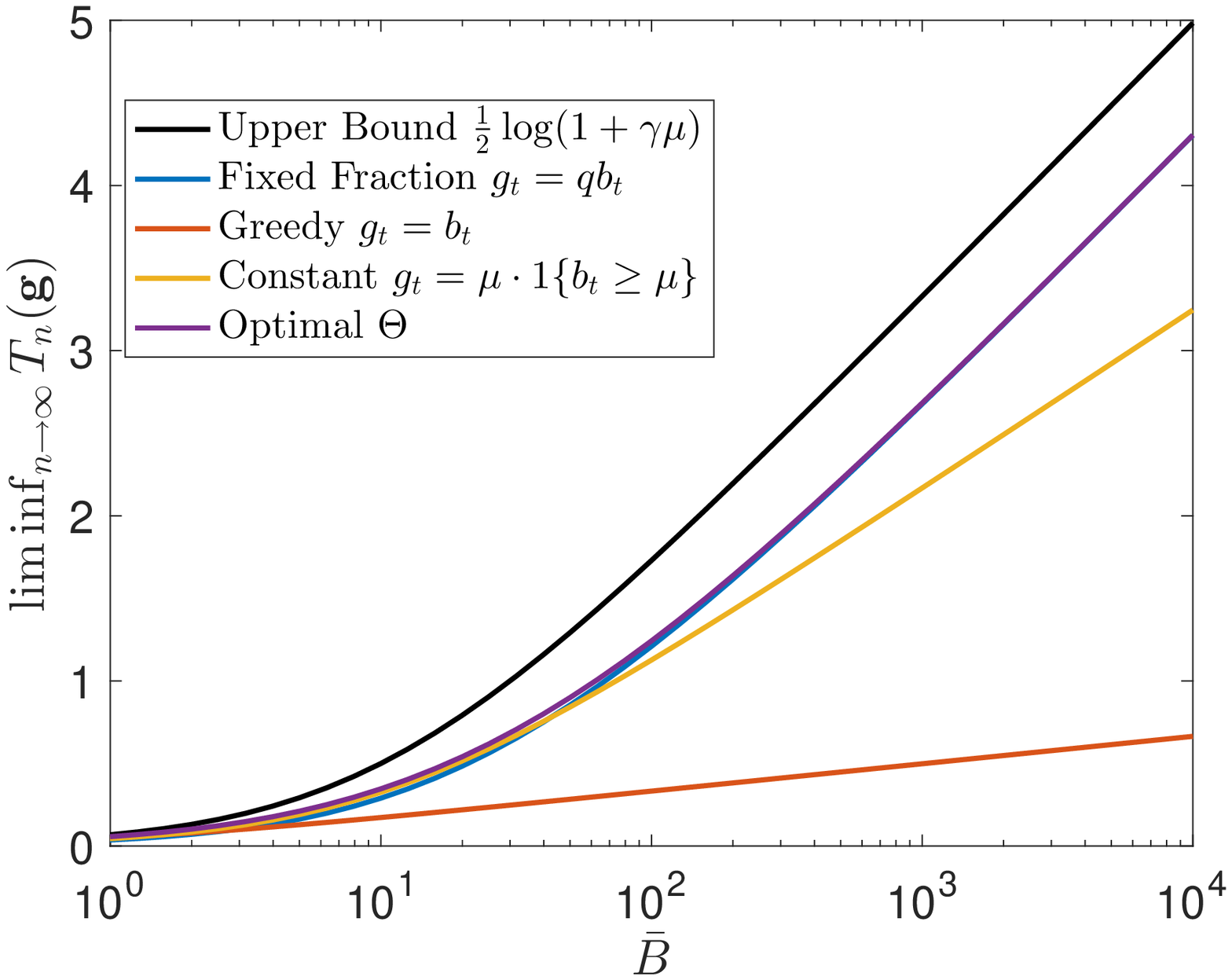}
\includegraphics[width=\smallfigwidth]{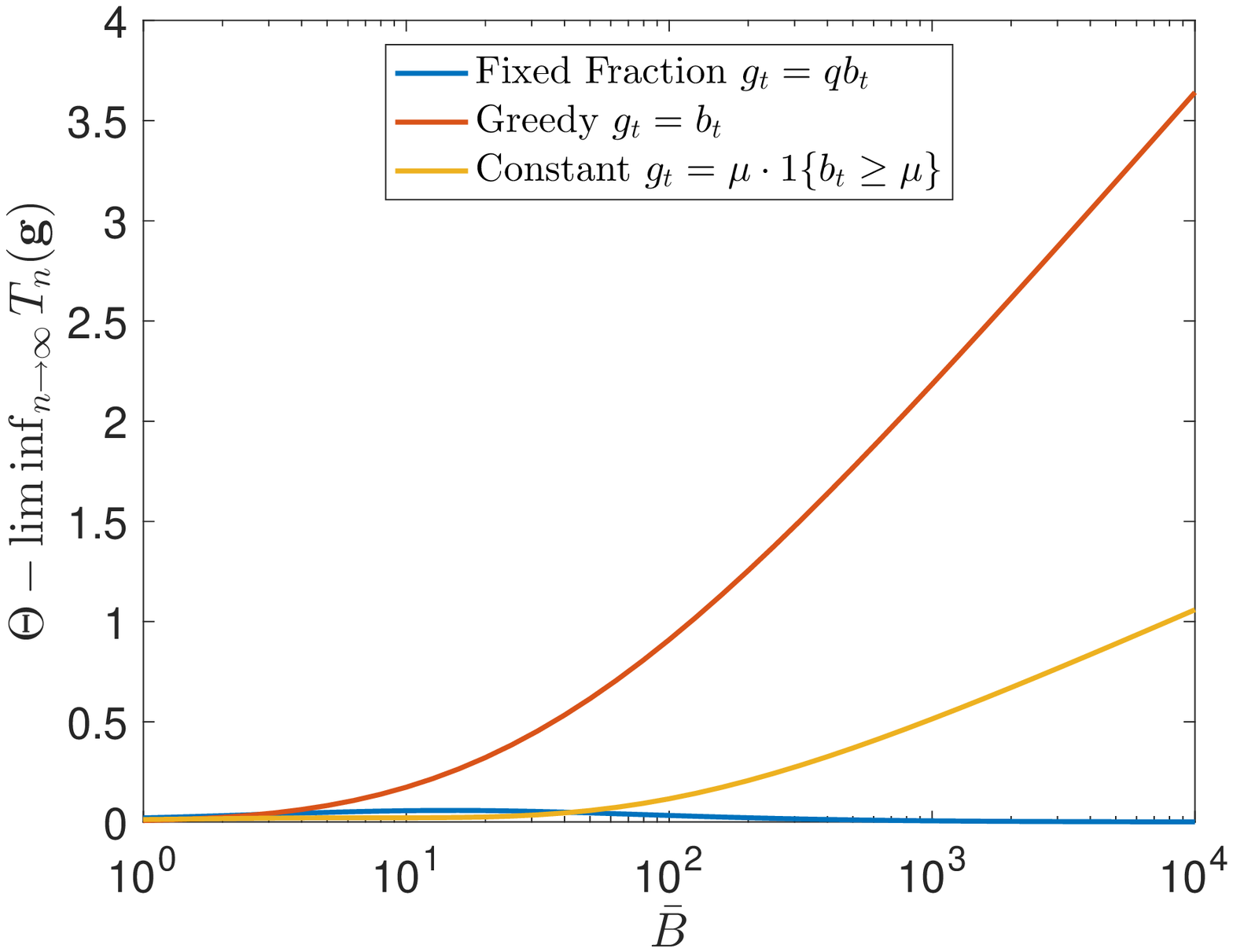}
\includegraphics[width=\smallfigwidth]{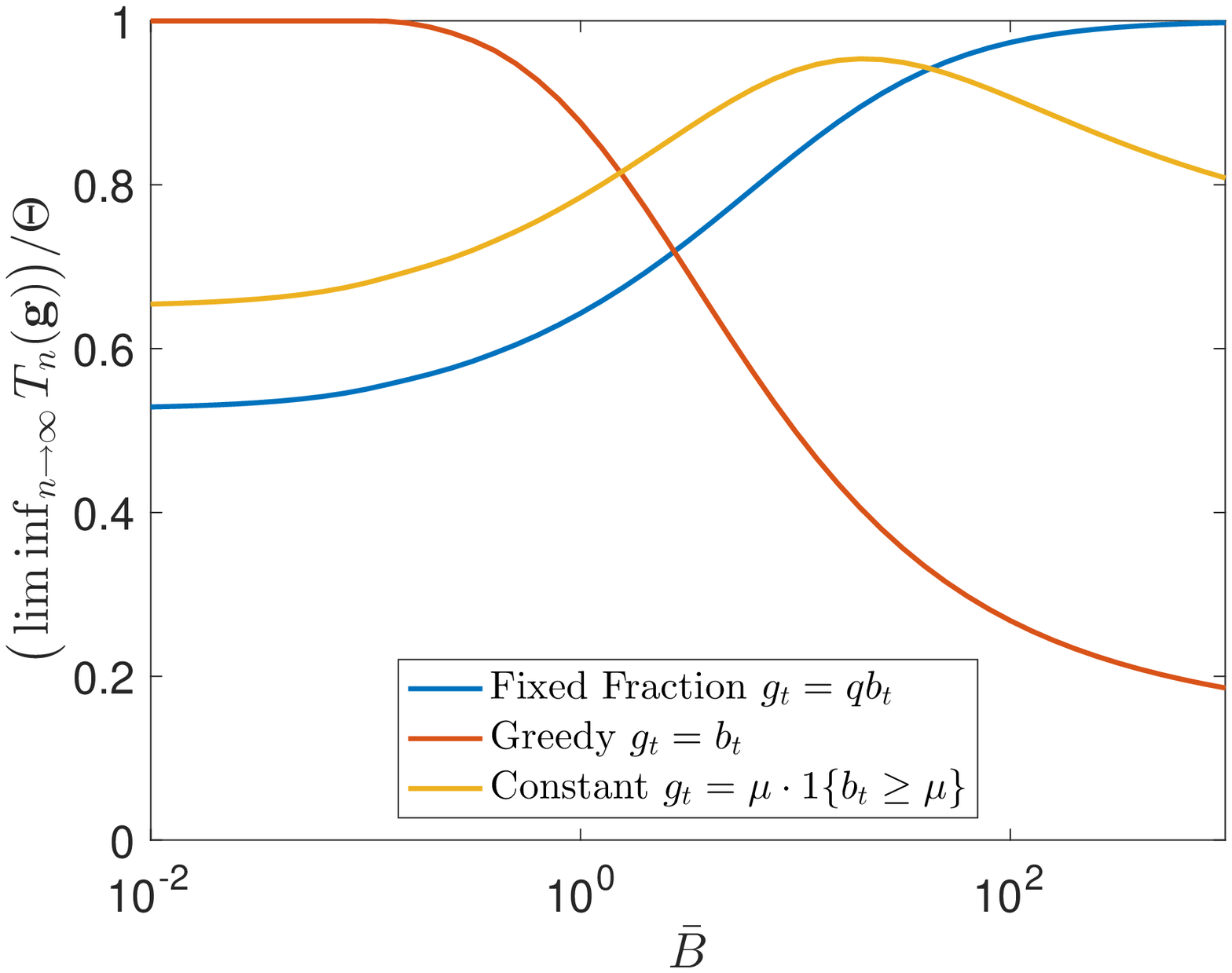}
\caption{Plots of the Fixed Fraction, Greedy, and Constant policies, for $\gamma=1$ and $E_t\sim \text{Bernoulli}(0.1)$, where $E_t\in\{0,\bar{B}\}$.}
\label{fig:sim_bernoulli0.1}
\vspace{2em}
\includegraphics[width=\smallfigwidth]{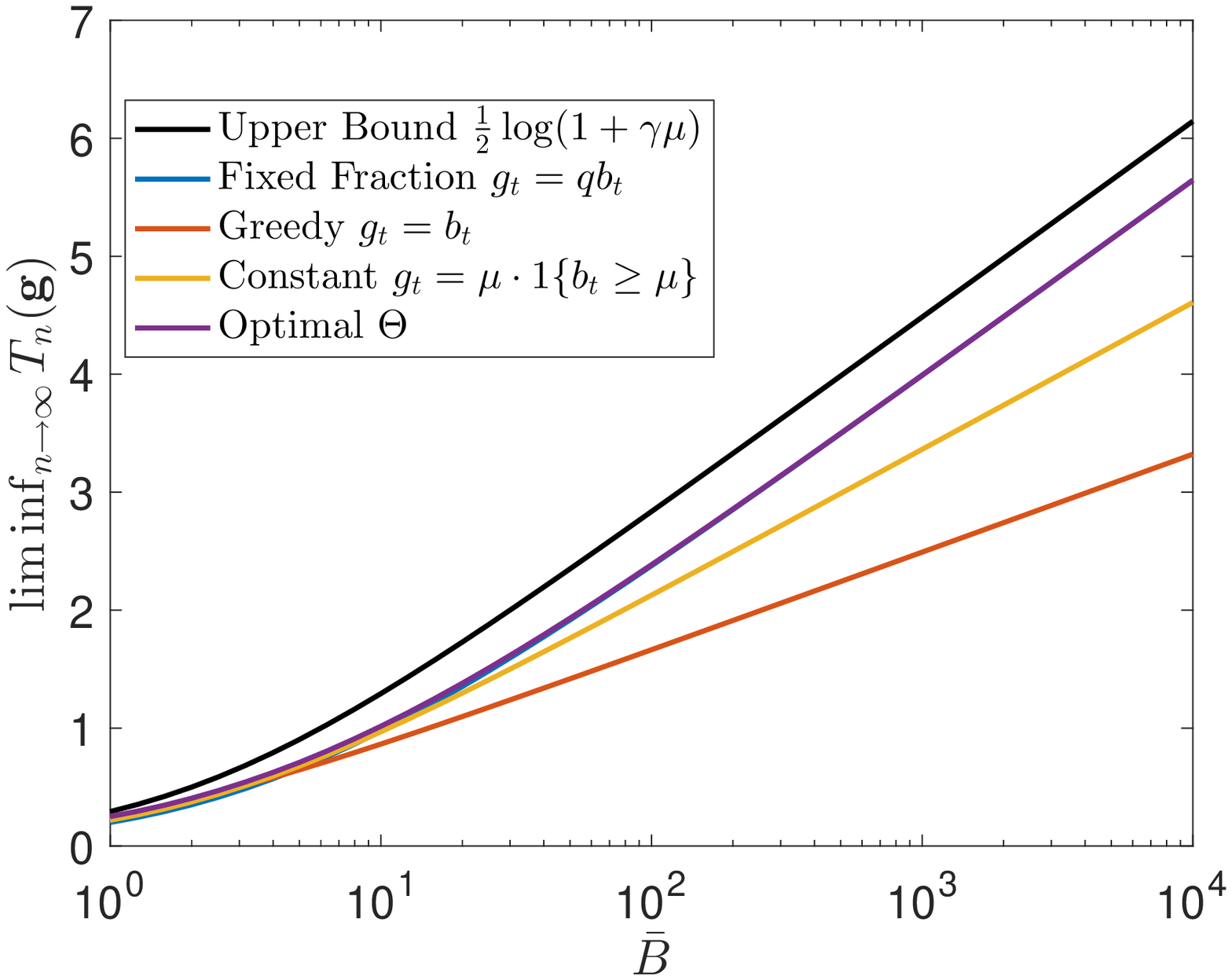}
\includegraphics[width=\smallfigwidth]{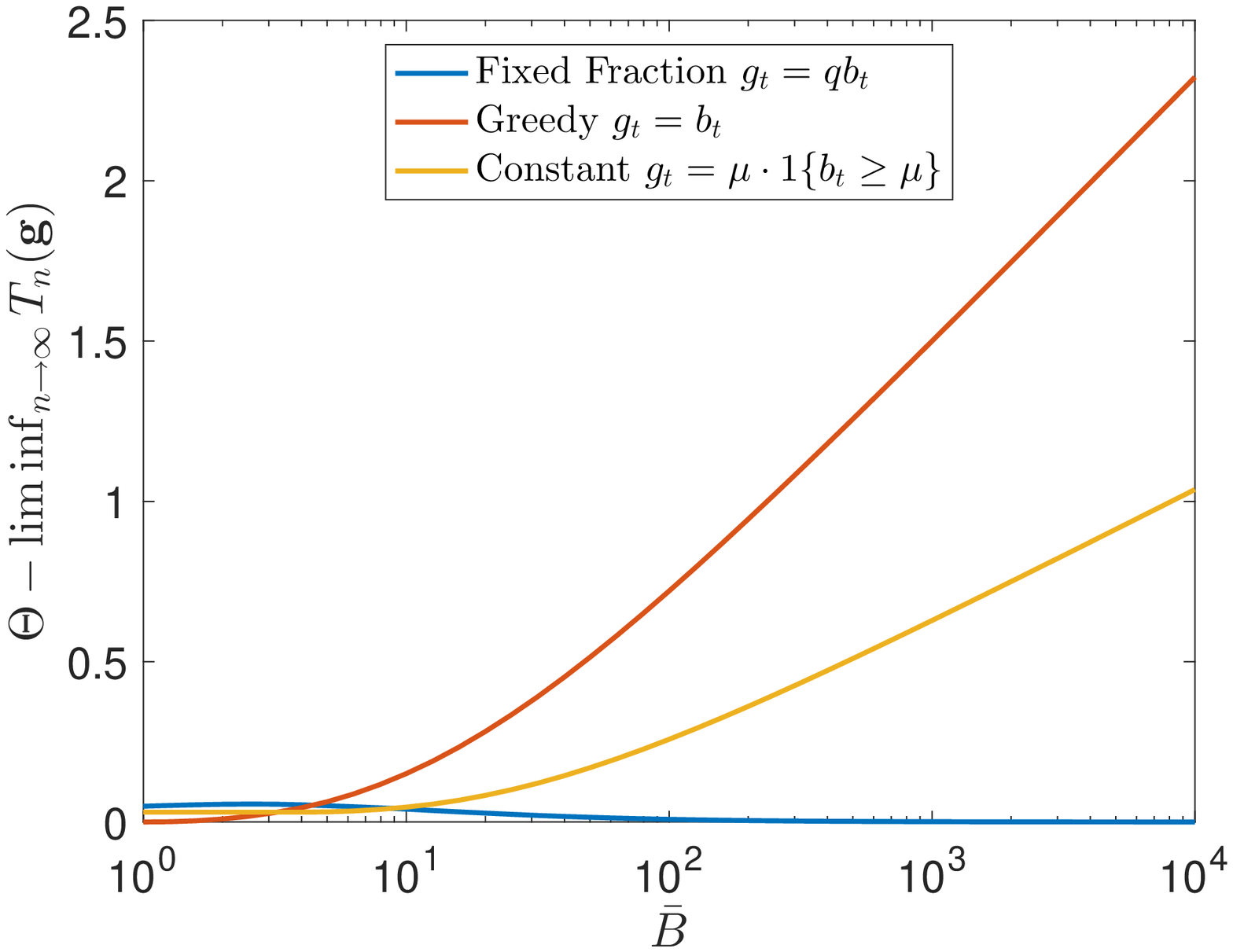}
\includegraphics[width=\smallfigwidth]{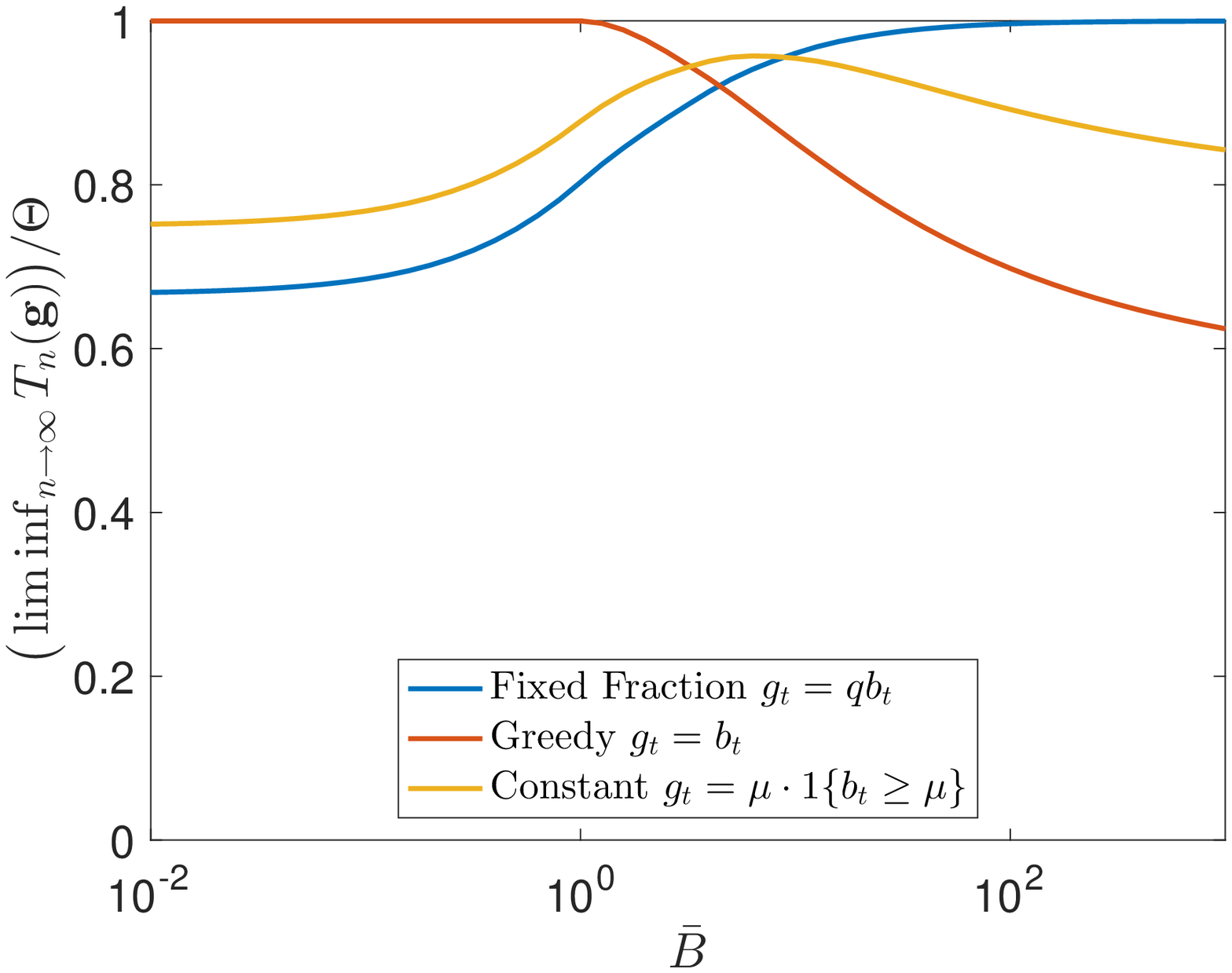}
\caption{Plots of the Fixed Fraction, Greedy, and Constant policies, for $\gamma=1$ and $E_t\sim \text{Bernoulli}(0.5)$, where $E_t\in\{0,\bar{B}\}$.}
\label{fig:sim_bernoulli0.5}
\vspace{2em}
\includegraphics[width=\smallfigwidth]{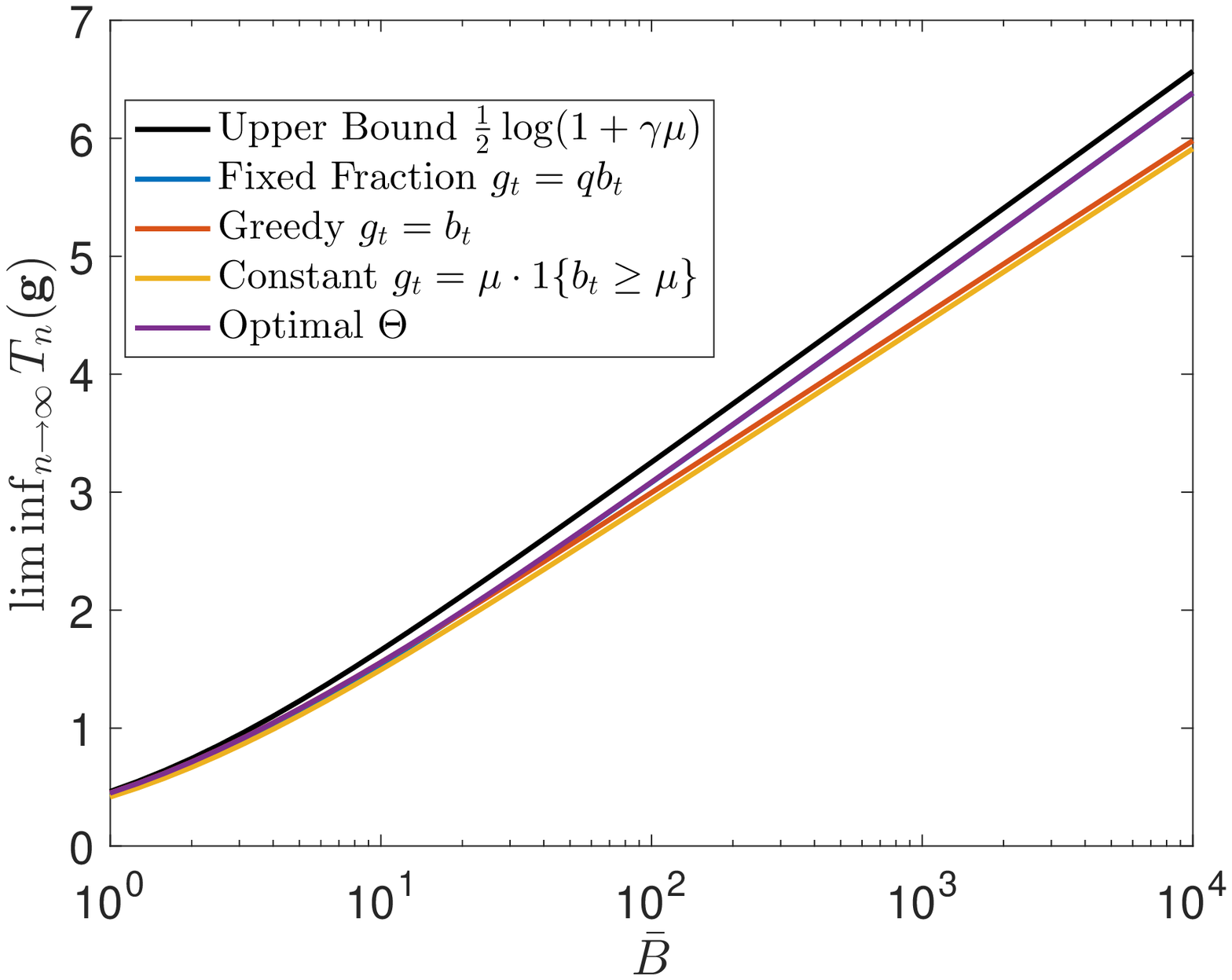}
\includegraphics[width=\smallfigwidth]{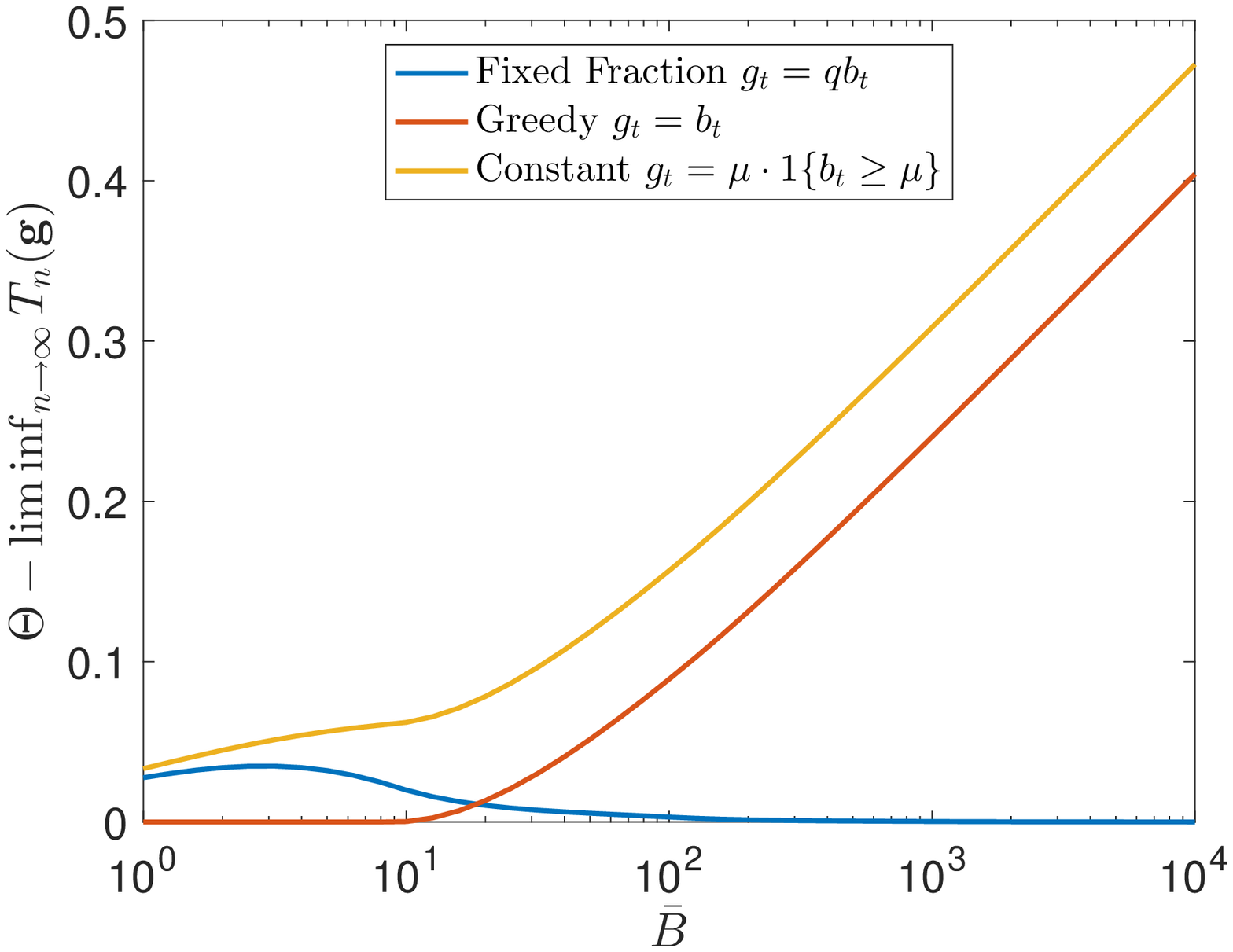}
\includegraphics[width=\smallfigwidth]{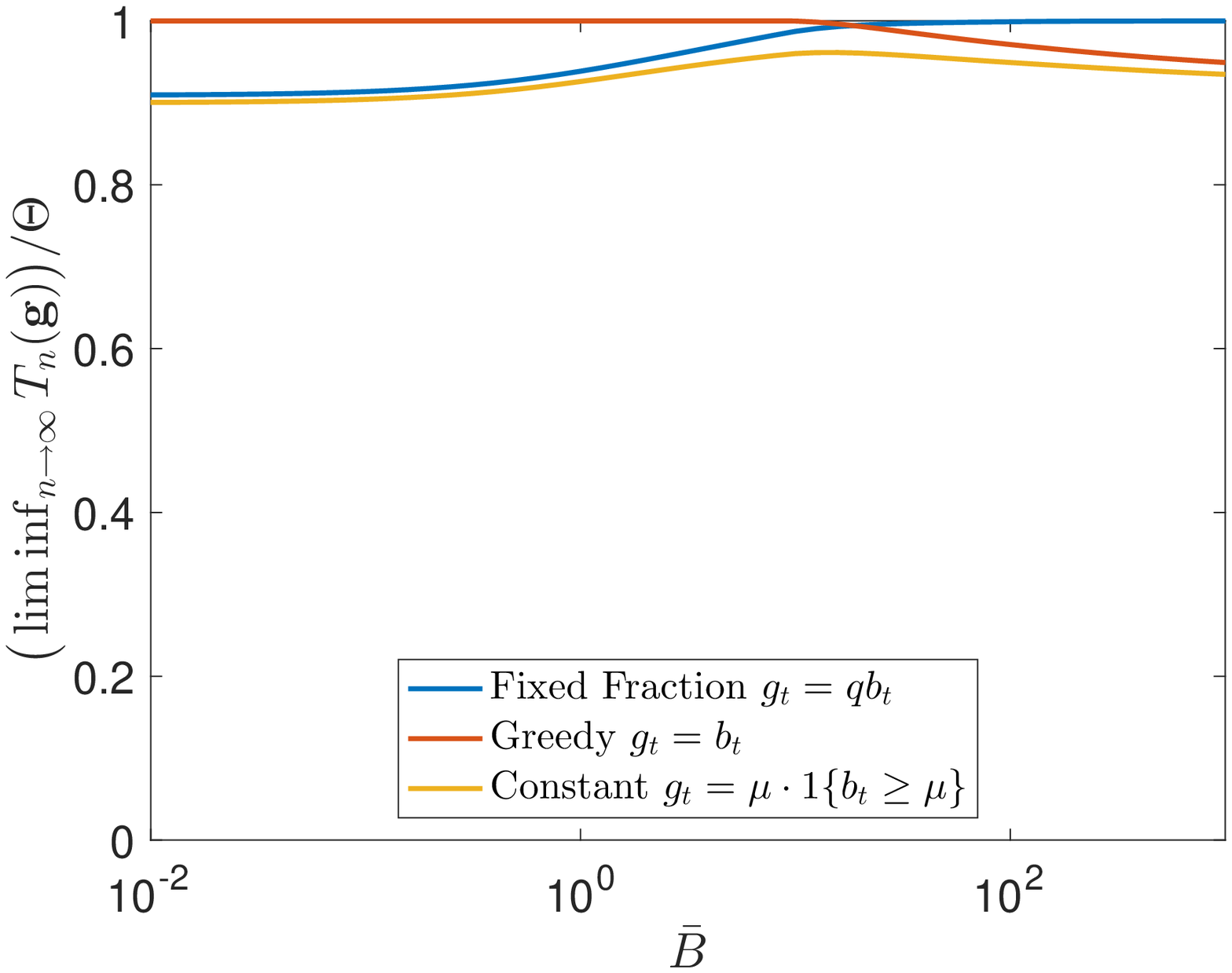}
\caption{Plots of the Fixed Fraction, Greedy, and Constant policies, for $\gamma=1$ and $E_t\sim \text{Bernoulli}(0.9)$, where $E_t\in\{0,\bar{B}\}$.}
\label{fig:sim_bernoulli0.9}
\end{figure*}
\begin{figure*}
\centering
\includegraphics[width=\smallfigwidth]{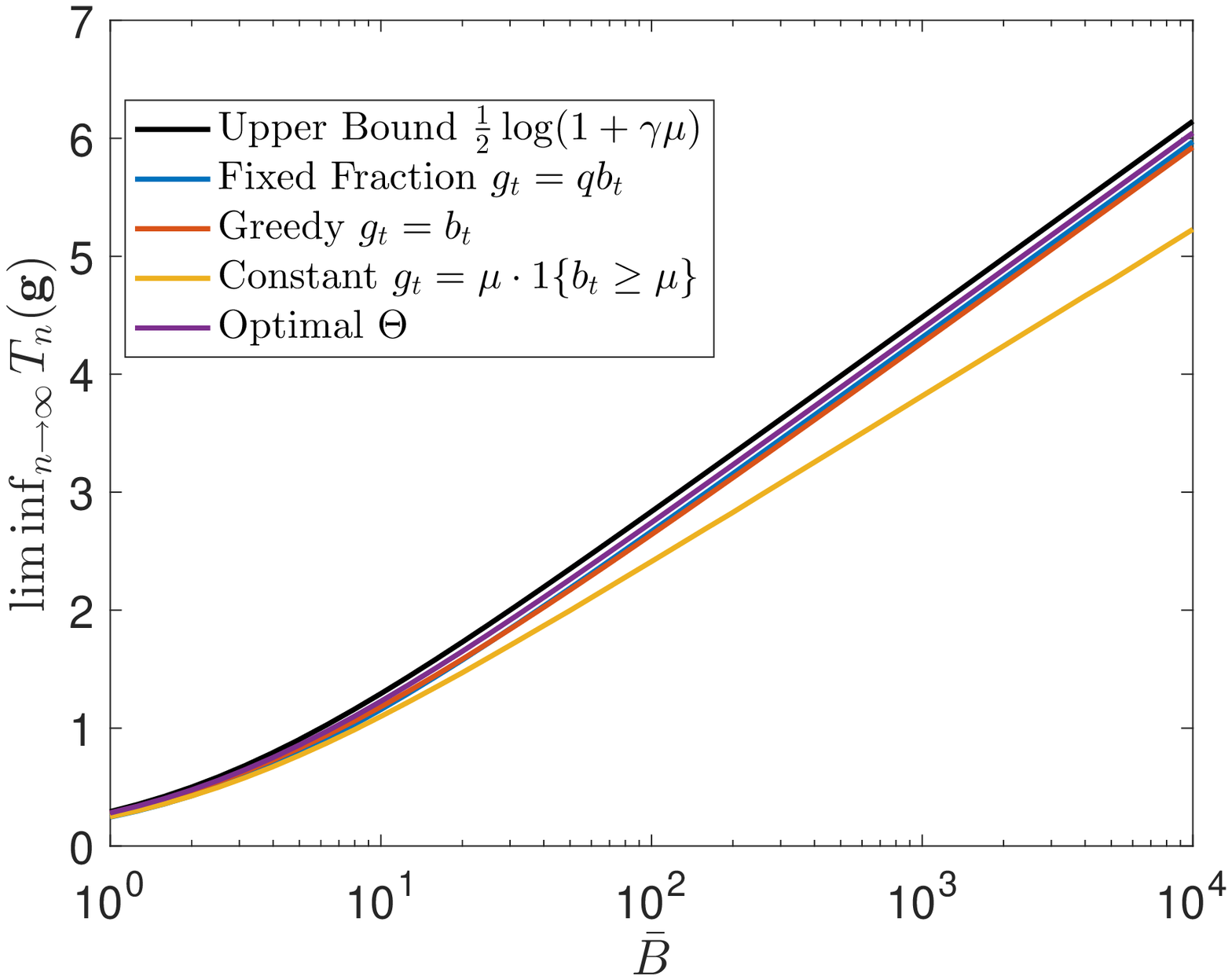}
\includegraphics[width=\smallfigwidth]{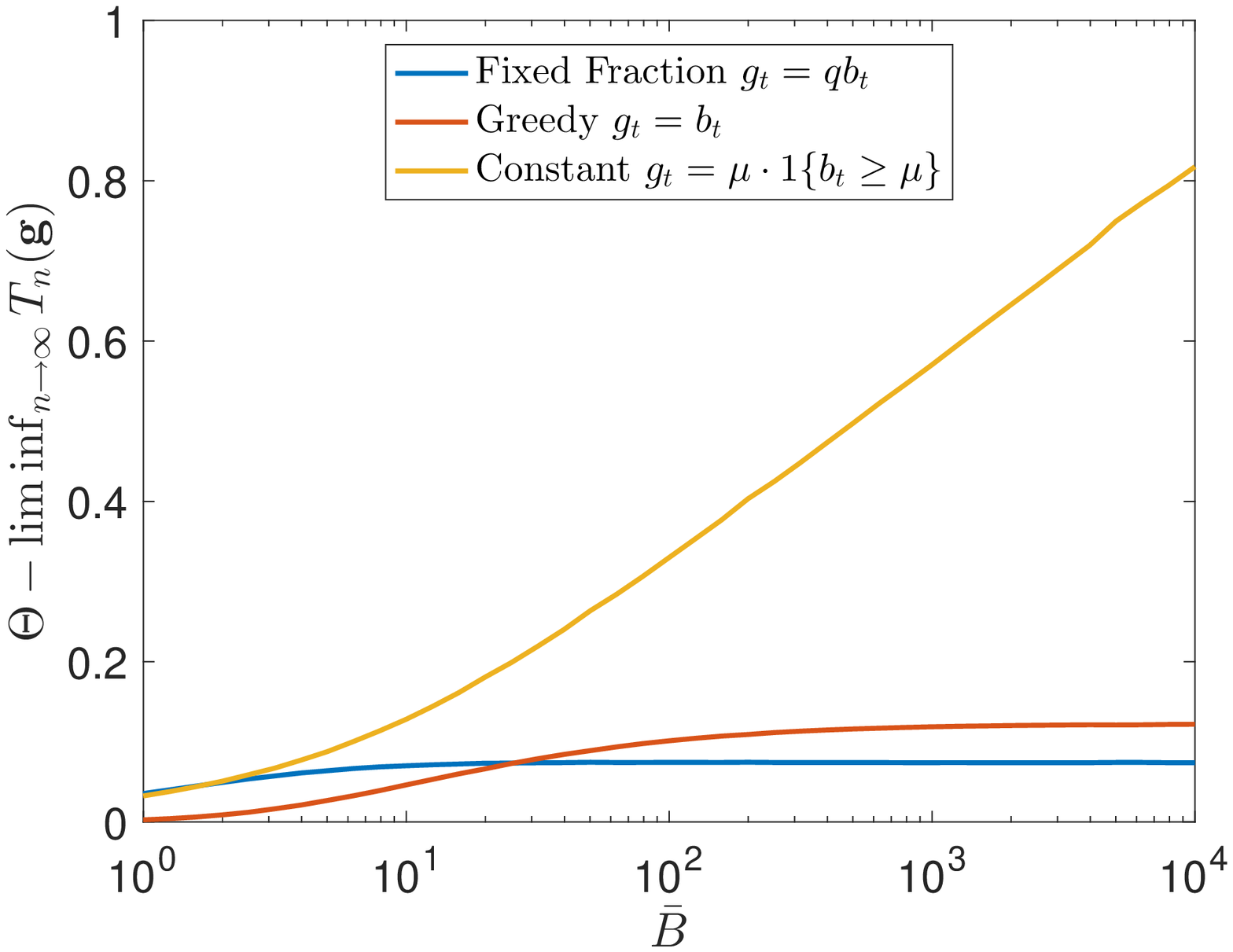}
\includegraphics[width=\smallfigwidth]{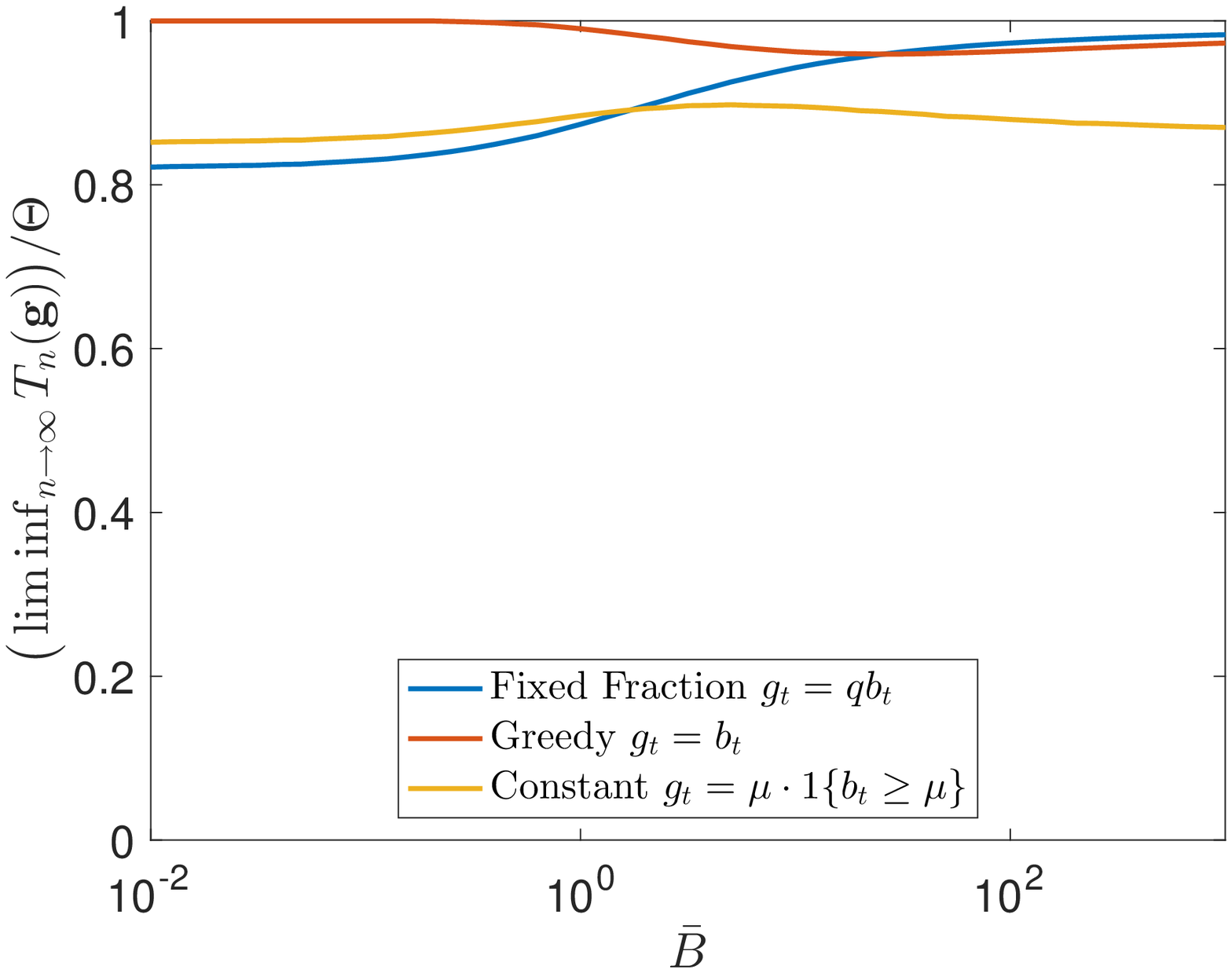}
\caption{Plots of the Fixed Fraction, Greedy, and Constant policies, for $\gamma=1$ and $E_t\sim \textrm{Unif}[0,\bar{B}]$.}
\label{fig:sim_uniform}
\vspace{2em}
\includegraphics[width=\smallfigwidth]{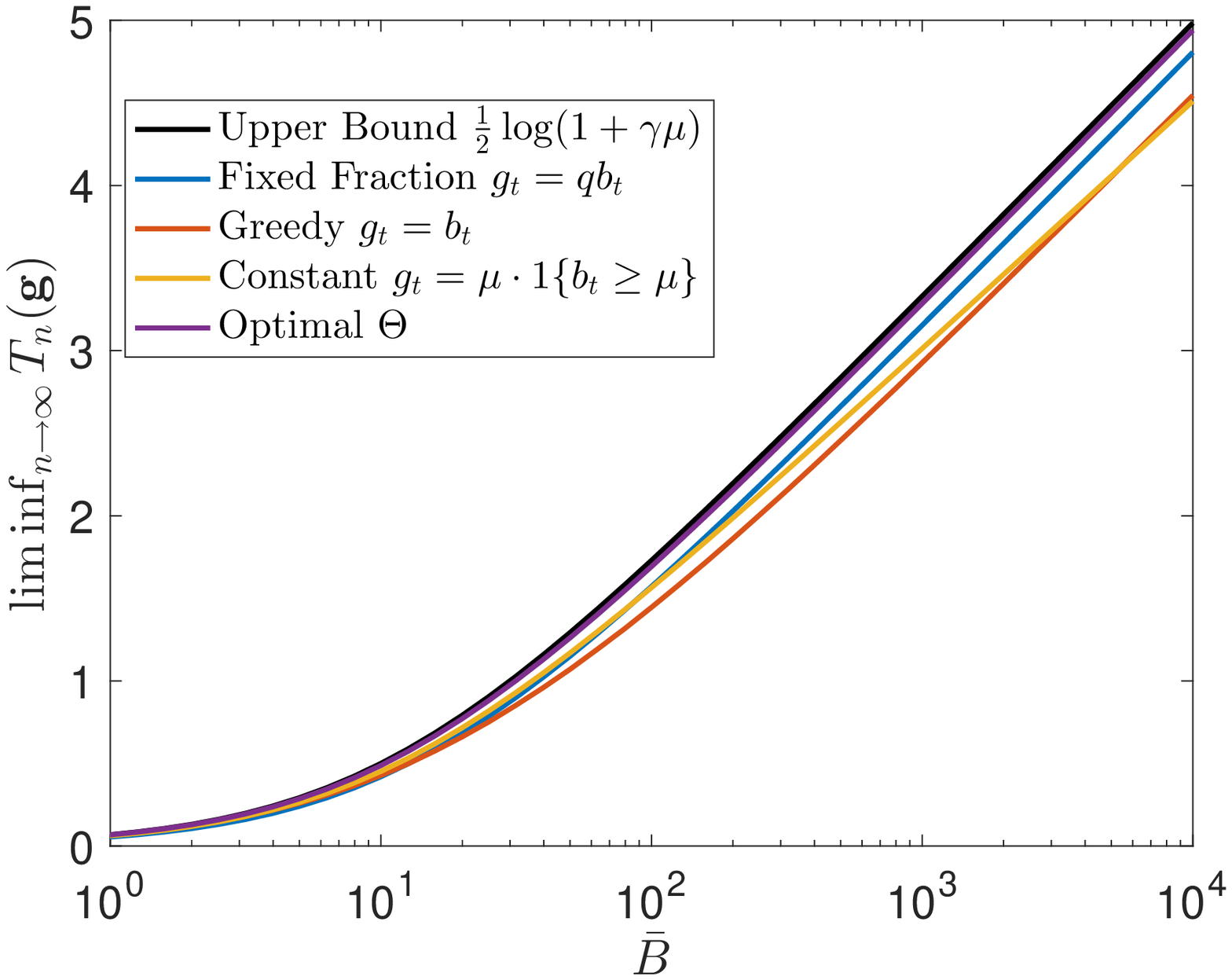}
\includegraphics[width=\smallfigwidth]{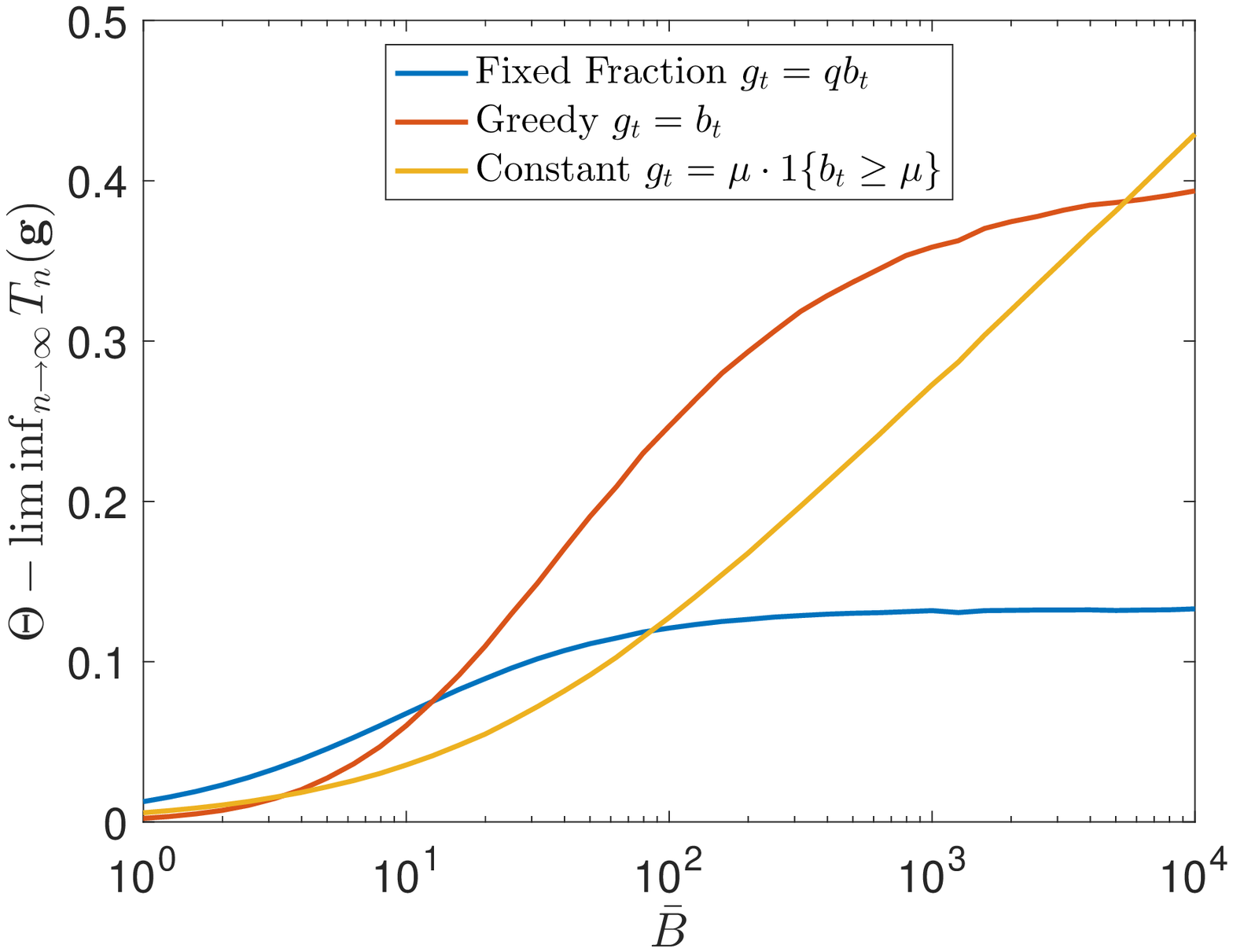}
\includegraphics[width=\smallfigwidth]{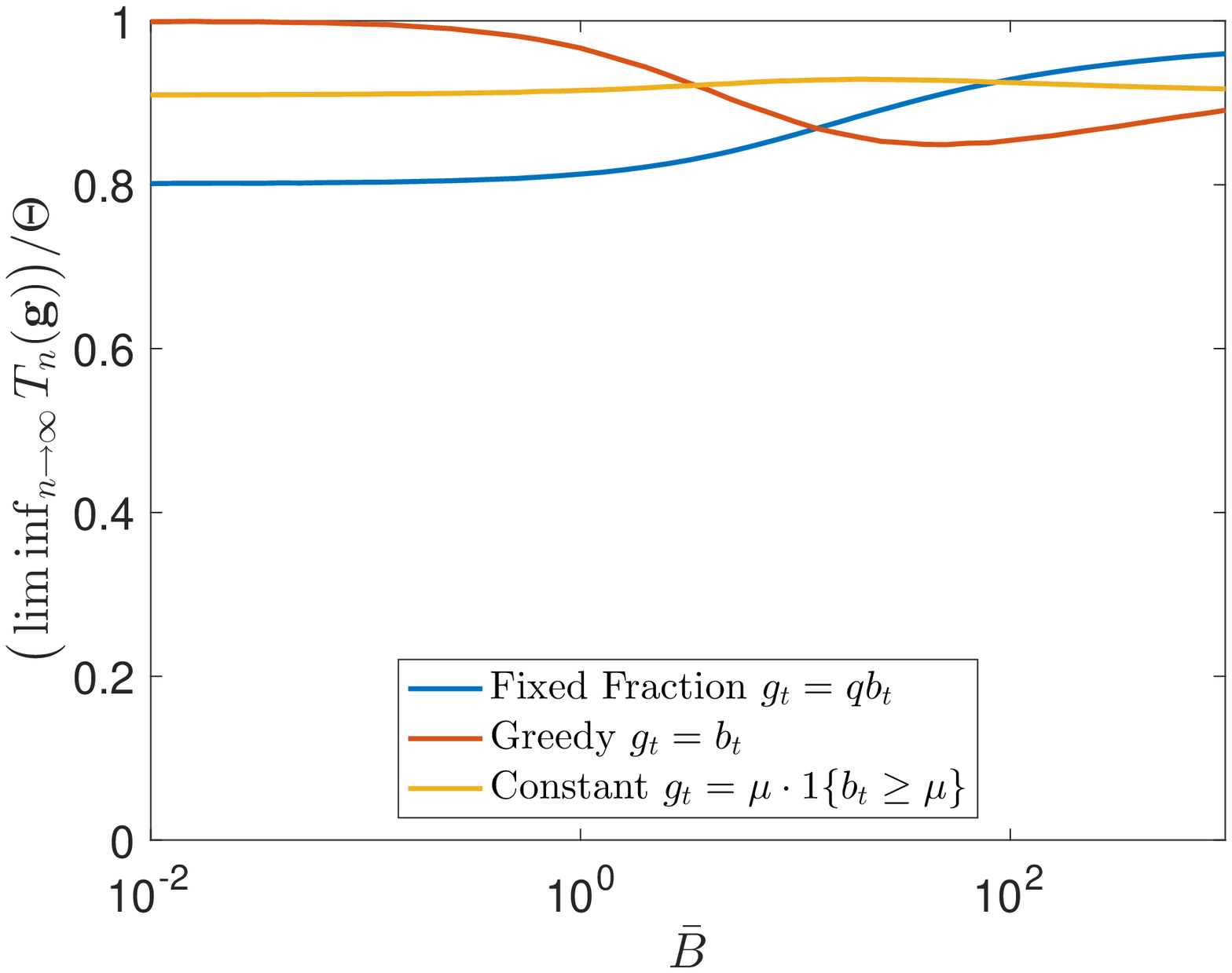}
\caption{Plots of the Fixed Fraction, Greedy, and Constant policies, for $\gamma=1$ and $E_t\sim \text{Exp}(\frac{1}{0.1\cdot\bar{B}})$.}
\label{fig:sim_exponential}
\end{figure*}
Figures~\ref{fig:sim_uniform} and~\ref{fig:sim_exponential} provide the corresponding plots for the uniform and the exponential distribution. The exponential distribution corresponds to an energy harvesting process arising from a Gaussian signal. Note that while the absolute gap to optimality for the three policies depend on the distribution, the general trends for the Bernoulli case prevail.

\section{Lower Bounding the Throughput Achieved by the Fixed Fraction Policy}
\label{sec:throughput_lower_bounds}

In this section we will prove the main theorem of the paper, namely Theorem \ref{thm:onlinePC}.
First, we will derive lower bounds on the throughput obtained by the Fixed Fraction Policy when the energy arrivals are i.i.d. Bernoulli. These are derived in Sections \ref{subsec:Bernoulli} and \ref{subsec:Bernoulli_mult}, in the form of an additive gap and a multiplicative gap from optimality, respectively.
Next, in Section \ref{subsec:general_energy_arrivals} we will show that the throughput of the Fixed Fraction Policy under \emph{any} i.i.d. energy arrival process is necessarily larger than under Bernoulli energy arrivals. This will imply that the lower bounds derived for the Bernoulli case, apply also to any harvesting process.

\subsection{Additive Gap for Bernoulli Energy Arrivals: Proof of Proposition~\ref{prop:Bernoulli_gap}}
\label{subsec:Bernoulli}

Before moving forward to establish the approximate optimality of this power control policy, we provide a few definitions and results from renewal theory.
\begin{definition}
\label{def:regenerative}
A stochastic process $\{X_t\}_{t=1}^{\infty}$ is called a \emph{non-delayed regenerative process} if there exists a random time $\tau>0$ such that the process $\{X_{\tau+t}\}_{t=1}^{\infty}$ has the same distribution as $\{X_t\}_{t=1}^{\infty}$ and is independent of the past $(\tau,X^{\tau})$.
\end{definition}
Observe that a regenerative process is composed of i.i.d. ``cycles'' or \emph{epochs}, which have i.i.d. durations $\tau_1,\tau_2,\ldots$. At the beginning of each epoch, the process ``regenerates'' and all memory of the past is essentially erased.
The following lemma establishes an important time-average property of regenerative processes.
\begin{lemma}[LLN for Regenerative Processes]
\label{lemma:SLLN_regenerative}
Let $\{X_t\}_{t=1}^{\infty}$, $X_t\in\mathcal{X}$, be a non-delayed regenerative process with associated epoch duration $\tau$, and let $f:\mathcal{X}\to\mathbb{R}$.
If $\mathbb{E}\tau<\infty$ and $\mathbb{E}[\sum_{t=1}^{\tau}|f(X_t)|]<\infty$ then:
\[
\lim_{n\to\infty}\frac{1}{n}\sum_{t=1}^{n}f(X_t)=\frac{1}{\mathbb{E}\tau}\mathbb{E}\left[\sum_{t=1}^{\tau}f(X_t)\right]
\quad\text{a.s.}
\]
\end{lemma}
This is an immediate consequence of Theorem 3.1 in~\cite[Ch.~VI]{asmussen2008applied} or of the renewal reward theorem~\cite[Prop.~7.3]{ross2014introduction}.

We now return to our throughput optimization problem with Bernoulli energy arrivals.
Denote by $L$ the random time between two consecutive energy arrivals, or the length of an \emph{epoch}. Evidently, $L\sim\mathrm{Geometric}(p)$. That is,
\[
\Pr(L=k)=p(1-p)^{k-1}
\qquad,k=1,2,\ldots
\]
Recall that we assume, without loss of generality, that $b_1=\bar{B}$
(It is shown in Appendix~\ref{sec:initial_battery_state} that
the initial battery level is irrelevant to the long-term average throughput
 -- this follows from the fact that we can always wait until the battery recharges to $\bar{B}$ before starting transmission, with a vanishing penalty to the average throughput).

Equipped with Lemma~\ref{lemma:SLLN_regenerative}, we consider the Fixed Fraction Policy $\mathbf{g}$ for the Bernoulli case (see Section~\ref{sec:bernoulli_energy_arrivals}). Note that in this case this policy reduces to \eqref{Bernoulliform1} (or equivalently~\eqref{Bernoulliform2}).
Observe that $g_t(E^t)$ is a non-delayed regenerative process with epoch duration $L$. 
We apply Lemma~\ref{lemma:SLLN_regenerative} with $f(x)=\frac{1}{2}\log(1+x)$.
Note that $\mathbb{E}L=1/p<\infty$ and 
\[\mathbb{E}\Big[\sum_{t=1}^{L}|\tfrac{1}{2}\log(1+g_t(E^t)|\Big]\leq\mathbb{E}[L\cdot\tfrac{1}{2}\log(1+\bar{B})]<\infty,\] so the conditions of the lemma are satisfied. We obtain
\begin{align}
&\lim_{n\to\infty}\frac{1}{n}\sum_{t=1}^{n}\frac{1}{2}\log\big(1+g_t(E^t)\big)\nonumber\\*
&\hspace{4em}
=\frac{1}{\mathbb{E}L}\mathbb{E}\left[\sum_{t=1}^{L}\frac{1}{2}\log\big(1+g_t(E^t)\big)\right]\quad\text{a.s.}
\label{eq:SLLN_power_control_policy}
\end{align}

We proceed to lower bound the average throughput obtained by our suggested power control policy:
\begin{align}
\hspace{1em}\lefteqn{\hspace{-1em}\liminf_{n\to\infty}\T_n(\mathbf{g})}\nonumber\\*
&=\liminf_{n\to\infty}\frac{1}{n}\sum_{t=1}^n \mathbb{E}\left[\frac{1}{2}\log (1+\gamma g_t(E^t))\right]\nonumber\\*
&\overset{\text{(i)}}{\geq}\mathbb{E}\left[\liminf_{n\to\infty}\frac{1}{n}\sum_{t=1}^{n}\frac{1}{2}\log(1+\gamma g_t(E^t))\right]\nonumber\\
&\overset{\text{(ii)}}{=}\mathbb{E}\left[\frac{1}{\mathbb{E}L}\mathbb{E}\left[\sum_{t=1}^{L}\frac{1}{2}\log(1+\gamma g_t(E^t))\right]\right]\nonumber\\
&=\frac{1}{\mathbb{E}L}\mathbb{E}\left[\sum_{t=1}^{L}\frac{1}{2}\log(1+\gamma g_t(E^t))\right]\nonumber\\
&\overset{\text{(iii)}}{=}\frac{1}{\mathbb{E}L}\mathbb{E}\left[\sum_{i=1}^{L}\frac{1}{2}\log(1+\gamma \bar{B}p(1-p)^{i-1})\right]\label{eq:Bernoulli_scheme_exact_TP}\\
&\overset{\text{(iv)}}{\geq}\frac{1}{\mathbb{E}L}\mathbb{E}\left[\sum_{i=1}^{L}\left(\frac{1}{2}\log(1+\gamma p\bar{B})+(i-1)\frac{1}{2}\log(1-p)\right)\right]\nonumber\\
&=\frac{1}{\mathbb{E}L}\mathbb{E}\left[L\frac{1}{2}\log(1+\gamma p\bar{B})+\frac{L(L-1)}{2}\frac{1}{2}\log(1-p)\right]\nonumber\\
&=\frac{1}{2}\log(1+\gamma p\bar{B})-\frac{1}{4}\left(\frac{\mathbb{E}[L^2]}{\mathbb{E}L}-1\right)\log\left(\frac{1}{1-p}\right)\nonumber\\
&\overset{\text{(v)}}{=}\frac{1}{2}\log(1+\gamma p\bar{B})-\frac{1-p}{2p}\log\left(\frac{1}{1-p}\right),
\label{eq:Bernoulli_first_lower_bound}
\end{align}
where (i) is by Fatou's lemma~\cite[Theorem~1.5.4]{durrett2010probability}; (ii) is due to~\eqref{eq:SLLN_power_control_policy}; (iii) is by definition of the Fixed Fraction policy; (iv) is due to the inequality $\log(1+\alpha x)\geq\log(1+x)+\log\alpha$ for $0<\alpha\leq 1$;
and (v) is because $L\sim\mathrm{Geometric}(p)$.

The second term in \eqref{eq:Bernoulli_first_lower_bound} achieves its maximum when $p\to0$, in which case it is given by $\frac{1}{2\ln 2}\approx 0.72$. We conclude that for Bernoulli energy arrivals:
\begin{equation}
\label{eq:Bernoulli_lower_bound}
\liminf_{n\to\infty}\T_n(\mathbf{g})
\geq\frac{1}{2}\log(1+\gamma \mu)-\frac{1}{2\ln 2},
\end{equation}
where $\mu=p\bar{B}$ is the average energy arrival rate of the Bernoulli process.\qed


\subsection{Multiplicative Gap for Bernoulli Energy Arrivals: Proof of Proposition~\ref{prop:Bernoulli_mult}}
\label{subsec:Bernoulli_mult}

We start from \eqref{eq:Bernoulli_scheme_exact_TP}, which was derived in the previous section:
\begin{align}
\hspace{2em}\lefteqn{\hspace{-2em}\liminf_{n\to\infty}\T_n(\mathbf{g})}\nonumber\\*
&\geq\frac{1}{\mathbb{E}L}\mathbb{E}\left[\sum_{i=1}^{L}\frac{1}{2}\log(1+\gamma \bar{B}p(1-p)^{i-1})\right]\nonumber\\
&\overset{\text{(i)}}{\geq}\frac{1}{\mathbb{E}L}\mathbb{E}
	\left[\sum_{i=1}^{L}(1-p)^{i-1}\frac{1}{2}
	\log(1+\gamma\bar{B}p)\right]\nonumber\\
&\overset{\text{(ii)}}{=}p\sum_{k=1}^{\infty}p(1-p)^{k-1}
	\sum_{i=1}^{k}(1-p)^{i-1}\frac{1}{2}
	\log(1+\gamma\bar{B}p)\nonumber\\
&=\sum_{k=1}^{\infty}p^2(1-p)^{k-1}\frac{1-(1-p)^k}{p}
	\frac{1}{2}\log(1+\gamma\bar{B}p)\nonumber\\
&=\frac{1}{2-p}\cdot\frac{1}{2}
	\log(1+\gamma\bar{B}p)\nonumber\\
&\overset{\text{(iii)}}{\geq}\frac{1}{2}\cdot
	\frac{1}{2}\log(1+\gamma\mu),\label{eq:Bernoulli_lower_bound_mult}
\end{align}
where (i) is by the inequality $\log(1+\alpha x)\geq \alpha\log(1+x)$ for $0\leq\alpha\leq 1$;
(ii) is because $L\sim\text{Geometric}(p)$;
and (iii) is because $0\leq p\leq1$ and $\mu=p\bar{B}$.\qed

\subsection{General i.i.d Energy Harvesting Processes: Proof of \mbox{Theorem}~\ref{thm:onlinePC}}
\label{subsec:general_energy_arrivals}

We will now use the result of the previous section to lower bound the throughput of the Fixed Fraction Policy for general i.i.d. energy harvesting processes.
We will show that under all distributions of (clipped) energy arrivals $\tilde{E}_t=\min\{E_t,\bar{B}\}$ with mean $\mu=\mathbb{E}[\tilde{E}_t]$, the lowest throughput is obtained when $\tilde{E}_t$ is Bernoulli, taking the values $0$ or $\bar{B}$.

We begin with a few notations and definitions.
Recall that the Fixed Fraction Policy is given by $g_t=qb_t$, where $q=\mu/\bar{B}$.
Under this policy:
\begin{align*}
b_t=\min\{(1-q)b_{t-1}+\tilde{E}_t,\bar{B}\},
&&t=2,3,\ldots,
\end{align*}
where, as in the previous section, we assume $b_1=\bar{B}$.

In what follows, we consider the performance of this policy under different distributions  for the energy arrivals and different initial battery levels.
Therefore, with a slight abuse of notation, we define the expected $n$-horizon throughput for initial battery level $x\in[0,\bar{B}]$ and i.i.d. energy arrivals distributed according to the distribution of $E_1$:
\begin{align*}
\T_n(\mathbf{g},E_1,x)\triangleq
	\frac{1}{n}\sum_{t=1}^{n}\mathbb{E}[\tfrac{1}{2}\log(1+\gamma qb_t)
	|b_1=x].
\end{align*}
Note that the long term average throughput under i.i.d. energy arrivals with distribution $\tilde{E}_1$ is given by $\liminf_{n\to\infty}\T_n(\mathbf{g},\tilde{E}_1,\bar{B})$.

Let $\hat{E}_t$ be i.i.d. Bernoulli RVs, specifically $\hat{E}_t\in\{0,\bar{B}\}$ and $\Pr(\hat{E}_t=\bar{B})=q$. Note that 
\[\mathbb{E}[\hat{E}_t]=\mathbb{E}[\tilde{E}_t]=\mathbb{E}[\min\{E_t,\bar{B}\}]=\mu.\]
In the following proposition, we claim that the $n$-horizon expected throughput for any distribution of i.i.d. energy arrivals is always better than the throughput obtained for i.i.d. Bernoulli energy arrivals with the same mean, for any $n$ and any initial battery level $x$.
\begin{proposition}
\label{prop:bernoulli_is_worst}
For any $x\in[0,\bar{B}]$ and any integer $n\geq1$:
\[
\T_n(\mathbf{g},\tilde{E}_1,x)\geq\T_n(\mathbf{g},\hat{E}_1,x).
\]
\end{proposition}
Before proving this proposition, we state the following lemma.
\begin{lemma}
\label{lemma:bernoulli_is_worst}
Let $f(z)$ be concave on the interval $[0,\bar{B}]$, and let $Z$ be a RV confined to the same interval, i.e. $0\leq Z\leq \bar{B}$. 
Let $\hat{Z}\in\{0,\bar{B}\}$ be a Bernoulli RV with $\Pr(\hat{Z}=\bar{B})=\mathbb{E}Z/\bar{B}$.
Then
\[
\mathbb{E}[f(Z)]\geq \mathbb{E}[f(\hat{Z})].
\]
\end{lemma}
\begin{proof}
By concavity, for any $z\in[0,\bar{B}]$:
\[
f(z)\geq\frac{z}{\bar{B}}f(\bar{B})+\frac{\bar{B}-z}{\bar{B}}f(0).
\]
Setting $z=Z$ and taking expectation yields
\begin{align*}
\mathbb{E}[f(Z)]&\geq\frac{\mathbb{E}Z}{\bar{B}}f(\bar{B})
+\left(1-\frac{\mathbb{E}Z}{\bar{B}}\right)f(0)\\*
&=\mathbb{E}[f(\hat{Z})].
\tag*{\qedhere} 
\end{align*}
\end{proof}

\begin{proof}[Proof of Proposition~\ref{prop:bernoulli_is_worst}]
We will give a proof by induction.
Clearly for $n=1$ we have 
\[
\T_1(\mathbf{g},\tilde{E}_1,x)=\T_1(\mathbf{g},\hat{E}_1,x)
=\tfrac{1}{2}\log(1+\gamma qx).
\]
Observe that this is a non-decreasing concave function of $x$.
This will in fact be true for every $\T_n(\mathbf{g},\hat{E}_1,x)$, $n\geq 1$, and we will use this in the induction step.

Assume that $\T_{n-1}(\mathbf{g},\tilde{E}_1,x)\geq\T_{n-1}(\mathbf{g},\hat{E}_1,x)$ for all $x\in[0,\bar{B}]$, and also that $\T_{n-1}(\mathbf{g},\hat{E}_1,x)$ is monotonic non-decreasing and concave in $x$.

For the induction step, observe that:
\begin{align*}
n\T_n(\mathbf{g},\tilde{E}_1,x)
&=\tfrac{1}{2}\log(1+\gamma qx)
	+(n-1)\mathbb{E}[\T_{n-1}(\mathbf{g},\tilde{E}_1,b_2)],
\end{align*}
where the expectation is over the RV 
$b_2=\min\{(1-q)x+\tilde{E}_2,\bar{B}\}$.
This is due to the process $b_t$ being a time-homogeneous Markov chain.
By the induction hypothesis, we have:
\begin{equation}
n\T_n(\mathbf{g},\tilde{E}_1,x)
\geq\tfrac{1}{2}\log(1+\gamma qx)
	+(n-1)\mathbb{E}[\T_{n-1}(\mathbf{g},\hat{E}_1,b_2)],
\label{eq:induction_step}
\end{equation}
where still $b_2=\min\{(1-q)x+\tilde{E}_2,\bar{B}\}$.
Now,
\begin{align*}
\hspace{1ex}\lefteqn{\hspace{-1ex}\T_{n-1}(\mathbf{g},\hat{E}_1,b_2)}\nonumber\\*
&=\T_{n-1}(\mathbf{g},\hat{E}_1,\min\{(1-q)x+\tilde{E}_2,\bar{B}\})\\
&=\min\big\{\T_{n-1}(\mathbf{g},\hat{E}_1,(1-q)x+\tilde{E}_2),\ 
	\T_{n-1}(\mathbf{g},\hat{E}_1,\bar{B})\big\},
\end{align*}
where the second equality is because $\T_{n-1}(\mathbf{g},\hat{E}_1,{}\cdot{})$ is non-decreasing, due to the induction hypothesis.
Next, we claim that the function $f_1(z)\triangleq\T_{n-1}(\mathbf{g},\hat{E}_1,(1-q)x+z)$ is concave. This is true again by the induction hypothesis that $\T_{n-1}(\mathbf{g},\hat{E}_1,{}\cdot{})$ is concave.
Therefore, since $\T_{n-1}(\mathbf{g},\hat{E}_1,\bar{B})$ is simply a constant, the function $f_2(z)\triangleq\T_{n-1}(\mathbf{g},\hat{E}_1,\min\{(1-q)x+z,\bar{B}\})$ is a minimum of two concave functions, hence it is itself concave.
We can now apply Lemma~\ref{lemma:bernoulli_is_worst} to obtain:
\begin{align*}
\mathbb{E}[\T_{n-1}(\mathbf{g},\hat{E}_1,b_2)]
&=\mathbb{E}[f_2(\tilde{E}_2)]\\
&\geq\mathbb{E}[f_2(\hat{E}_2)]\\
&= \mathbb{E}[\T_{n-1}(\mathbf{g},\hat{E}_1,\hat{b}_2)],
\end{align*}
where $\hat{b}_2\triangleq\min\{(1-q)x+\hat{E}_2,\bar{B}\}$.
Substituting this into~\eqref{eq:induction_step}:
\begin{align*}
n\T_{n}(\mathbf{g},\tilde{E}_1,x)
&\geq\tfrac{1}{2}\log(1+\gamma qx)
	+(n-1)\mathbb{E}[\T_{n-1}(\mathbf{g},\hat{E}_1,\hat{b}_2)]\\*
&=n\T_{n}(\mathbf{g},\hat{E}_1,x).
\end{align*}
It is left to verify that $\T_{n}(\mathbf{g},\hat{E}_1,x)$ is concave and non-decreasing in $x$.
Writing it explicitly:
\begin{align*}
n\T_{n}(\mathbf{g},\hat{E}_1,x)
&=\tfrac{1}{2}\log(1+\gamma qx)
+q(n-1)\T_{n-1}(\mathbf{g},\hat{E}_1,\bar{B})\\*
&\quad{}+(1-q)(n-1)\T_{n-1}(\mathbf{g},\hat{E}_1,(1-q)x),
\end{align*}
we see that it is a sum of non-decreasing concave functions of $x$, hence it is a non-decreasing concave function of $x$.
\end{proof}

As an immediate result of Proposition~\ref{prop:bernoulli_is_worst}, we obtain
\begin{equation}
\liminf_{n\to\infty}\T_n(\mathbf{g},\tilde{E}_1,\bar{B})
\geq\liminf_{n\to\infty}\T_n(\mathbf{g},\hat{E}_1,\bar{B}).
\label{eq:bernoulli_is_worst}
\end{equation}
Now we can apply the results of the previous section.
From \eqref{eq:Bernoulli_lower_bound} we have
\begin{align*}
\liminf_{n\to\infty}\T_n(\mathbf{g},\hat{E}_1,\bar{B})
&\geq\frac{1}{2}\log(1+\gamma q\bar{B})-\frac{1}{\ln2}\\*
&=\frac{1}{2}\log(1+\gamma\mu)-0.72,
\end{align*}
and from \eqref{eq:Bernoulli_lower_bound_mult} we have
\begin{align*}
\liminf_{n\to\infty}\T_n(\mathbf{g},\hat{E}_1,\bar{B})
&\geq\frac{1}{2}\cdot\frac{1}{2}\log(1+\gamma q\bar{B})\\
&=\frac{1}{2}\cdot\frac{1}{2}\log(1+\gamma\mu).
\end{align*}
Substituting in~\eqref{eq:bernoulli_is_worst} completes the proof of Theorem~\ref{thm:onlinePC}.

\section{Conclusion}

We  proposed a simple online power control policy for energy harvesting channels and proved that it is within constant additive and multiplicative gaps to the AWGN capacity for any i.i.d. harvesting process and any battery size. This allowed us to develop a simple and insightful approximation for the optimal throughput. While optimal power control in the offline case and the online case with infinite battery size have been characterized in the previous literature, the strategies developed for the online case have been mostly heuristic with no or only asymptotic guarantees on optimality. We believe the approximation approach we propose in this paper can be fruitful in developing further insights on online power control under various assumptions, such as processing cost and battery non-idealities as well as multi-user settings, the rigorous treatment of which have been so far mostly limited to either the offline case or the case with infinite battery. For example, in~\cite{HuseyinWiOpt} this approach is extended to derive a universally near-optimal power control policy for the multiple-access channel. It is shown that in an energy-harvesting MAC the users can achieve a symmetric capacity equal to the AWGN capacity as $K\to\infty$.
Additionally, an approximation of the optimal throughput can be used to derive bounds on the information-theoretic capacity of the energy harvesting channel, as done in~\cite{ShavivNguyenOzgur2015}.

An important step in our proof was to show that i.i.d. Bernoulli energy arrivals constitute the worst case for our proposed policy among all i.i.d. energy arrival processes with the same mean, i.e. the throughput achieved by our proposed policy is smallest when the process is Bernoulli. Whether i.i.d. Bernoulli energy arrivals are also the worst case in terms of the optimal throughput is an interesting question.


\bibliographystyle{IEEEtran}
\bibliography{IEEEabrv,energy_harvesting}

\begin{thebibliography}{10}
\providecommand{\url}[1]{#1}
\csname url@samestyle\endcsname
\providecommand{\newblock}{\relax}
\providecommand{\bibinfo}[2]{#2}
\providecommand{\BIBentrySTDinterwordspacing}{\spaceskip=0pt\relax}
\providecommand{\BIBentryALTinterwordstretchfactor}{4}
\providecommand{\BIBentryALTinterwordspacing}{\spaceskip=\fontdimen2\font plus
\BIBentryALTinterwordstretchfactor\fontdimen3\font minus
  \fontdimen4\font\relax}
\providecommand{\BIBforeignlanguage}[2]{{%
\expandafter\ifx\csname l@#1\endcsname\relax
\typeout{** WARNING: IEEEtran.bst: No hyphenation pattern has been}%
\typeout{** loaded for the language `#1'. Using the pattern for}%
\typeout{** the default language instead.}%
\else
\language=\csname l@#1\endcsname
\fi
#2}}
\providecommand{\BIBdecl}{\relax}
\BIBdecl

\bibitem{YangUlukus2012}
J.~Yang and S.~Ulukus, ``Optimal packet scheduling in an energy harvesting
  communication system,'' \emph{{IEEE} Trans. Commun.}, vol.~60, no.~1, pp.
  220--230, 2012.

\bibitem{TutuncuogluYener2012}
K.~Tutuncuoglu and A.~Yener, ``Optimum transmission policies for battery
  limited energy harvesting nodes,'' \emph{{IEEE} Trans. Wireless Commun.},
  vol.~11, no.~3, pp. 1180--1189, 2012.

\bibitem{Ozeletal2011}
O.~Ozel, K.~Tutuncuoglu, J.~Yang, S.~Ulukus, and A.~Yener, ``Transmission with
  energy harvesting nodes in fading wireless channels: Optimal policies,''
  \emph{{IEEE} J. Sel. Areas Commun.}, vol.~29, no.~8, pp. 1732--1743, 2011.

\bibitem{DP0}
M.~Zafer and E.~Modiano, ``Optimal rate control for delay-constrained data
  transmission over a wireless channel,'' \emph{{IEEE} Trans. Inf. Theory},
  vol.~54, no.~9, pp. 4020--4039, 2008.

\bibitem{DP1}
C.~K. Ho and R.~Zhang, ``Optimal energy allocation for wireless communications
  powered by energy harvesters,'' in \emph{{IEEE} Int. Symp. Information Theory
  (ISIT)}, 2010, pp. 2368--2372.

\bibitem{DP2}
A.~Sinha and P.~Chaporkar, ``Optimal power allocation for a renewable energy
  source,'' in \emph{National Conference on Communications (NCC)}.\hskip 1em
  plus 0.5em minus 0.4em\relax IEEE, 2012, pp. 1--5.

\bibitem{DP3}
P.~Blasco, D.~Gunduz, and M.~Dohler, ``A learning theoretic approach to energy
  harvesting communication system optimization,'' \emph{{IEEE} Trans. Wireless
  Commun.}, vol.~12, no.~4, pp. 1872--1882, 2013.

\bibitem{ho2012optimal}
C.~K. Ho and R.~Zhang, ``Optimal energy allocation for wireless communications
  with energy harvesting constraints,'' \emph{{IEEE} Trans. Signal Process.},
  vol.~60, no.~9, pp. 4808--4818, 2012.

\bibitem{MaoCheungWong2012}
S.~Mao, M.~H. Cheung, and V.~W. Wong, ``An optimal energy allocation algorithm
  for energy harvesting wireless sensor networks,'' in \emph{{IEEE}
  International Conference on Communications (ICC)}, 2012, pp. 265--270.

\bibitem{Wang15}
X.~Wang, J.~Gong, C.~Hu, S.~Zhou, and Z.~Niu, ``Optimal power allocation on
  discrete energy harvesting model,'' \emph{EURASIP Journal on Wireless
  Communications and Networking}, vol. 2015, no.~1, pp. 1--14, 2015.

\bibitem{Kazerouni2015}
A.~Kazerouni and A.~{\"{O}}zg{\"{u}}r, ``Optimal online strategies for an
  energy harvesting system with {Bernoulli} energy recharges,'' in \emph{13th
  International Symposium on Modeling and Optimization in Mobile, Ad Hoc, and
  Wireless Networks (WiOpt)}, 2015, pp. 235--242.

\bibitem{Mitran}
M.~B. Khuzani and P.~Mitran, ``On online energy harvesting in multiple access
  communication systems,'' \emph{{IEEE} Trans. Inf. Theory}, vol.~60, no.~3,
  pp. 1883--1898, 2014.

\bibitem{online_infiniteB1}
C.~M. Vigorito, D.~Ganesan, and A.~G. Barto, ``Adaptive control of duty cycling
  in energy-harvesting wireless sensor networks,'' in \emph{4th Annual IEEE
  Communications Society Conference on Sensor, Mesh and Ad Hoc Communications
  and Networks (SECON'07)}, 2007, pp. 21--30.

\bibitem{online_infiniteB2}
V.~Sharma, U.~Mukherji, V.~Joseph, and S.~Gupta, ``Optimal energy management
  policies for energy harvesting sensor nodes,'' \emph{{IEEE} Trans. Wireless
  Commun.}, vol.~9, no.~4, pp. 1326--1336, 2010.

\bibitem{online_infiniteB3}
R.~Rajesh, V.~Sharma, and P.~Viswanath, ``Capacity of fading {Gaussian} channel
  with an energy harvesting sensor node,'' in \emph{IEEE Global
  Telecommunications Conference (GLOBECOM 2011)}, 2011, pp. 1--6.

\bibitem{online_infiniteB4}
R.~Srivastava and C.~E. Koksal, ``Basic performance limits and tradeoffs in
  energy-harvesting sensor nodes with finite data and energy storage,''
  \emph{IEEE/ACM Transactions on Networking (TON)}, vol.~21, no.~4, pp.
  1049--1062, 2013.

\bibitem{info2013}
Q.~Wang and M.~Liu, ``When simplicity meets optimality: Efficient transmission
  power control with stochastic energy harvesting,'' in \emph{Proc. IEEE
  INFOCOM}, 2013, pp. 580--584.

\bibitem{amirnavaei2015online}
F.~Amirnavaei and M.~Dong, ``Online power control strategy for wireless
  transmission with energy harvesting,'' in \emph{2015 IEEE 16th International
  Workshop on Signal Processing Advances in Wireless Communications (SPAWC)},
  2015, pp. 6--10.

\bibitem{xu2014throughput}
J.~Xu and R.~Zhang, ``Throughput optimal policies for energy harvesting
  wireless transmitters with non-ideal circuit power,'' \emph{{IEEE} J. Sel.
  Areas Commun.}, vol.~32, no.~2, pp. 322--332, 2014.

\bibitem{Puterman2005}
M.~L. Puterman, \emph{Markov Decision Processes: Discrete Stochastic Dynamic
  Programming}.\hskip 1em plus 0.5em minus 0.4em\relax John Wiley \& Sons,
  2005.

\bibitem{arapostathis1993discrete}
A.~Arapostathis, V.~S. Borkar, E.~Fern{\'a}ndez-Gaucherand, M.~K. Ghosh, and
  S.~I. Marcus, ``Discrete-time controlled {Markov} processes with average cost
  criterion: a survey,'' \emph{SIAM Journal on Control and Optimization},
  vol.~31, no.~2, pp. 282--344, 1993.

\bibitem{Bertsekas2001vol2}
D.~P. Bertsekas, \emph{Dynamic Programming and Optimal Control}, 2nd~ed.\hskip
  1em plus 0.5em minus 0.4em\relax Athena Scientific, 2001, vol.~2.

\bibitem{ShavivOzgur2015}
D.~Shaviv and A.~{\"{O}}zg{\"{u}}r, ``Capacity of the {AWGN} channel with
  random battery recharges,'' in \emph{{IEEE} Int. Symp. Information Theory
  (ISIT)}, 2015, pp. 136--140.

\bibitem{DongFarniaOzgur2015}
Y.~Dong, F.~Farnia, and A.~{\"{O}}zg{\"{u}}r, ``Near optimal energy control and
  approximate capacity of energy harvesting communication,'' \emph{{IEEE} J.
  Sel. Areas Commun.}, vol.~33, no.~3, pp. 540--557, 2015.

\bibitem{asmussen2008applied}
S.~Asmussen, \emph{Applied probability and queues}.\hskip 1em plus 0.5em minus
  0.4em\relax Springer Science \& Business Media, 2008, vol.~51.

\bibitem{ross2014introduction}
S.~M. Ross, \emph{Introduction to probability models}.\hskip 1em plus 0.5em
  minus 0.4em\relax Academic press, 2014.

\bibitem{durrett2010probability}
R.~Durrett, \emph{Probability: theory and examples}.\hskip 1em plus 0.5em minus
  0.4em\relax Cambridge university press, 2010.

\bibitem{HuseyinWiOpt}
H.~A. Inan and A.~{\"{O}}zg{\"{u}}r, ``Online power control for the energy
  harvesting multiple access channel,'' presented at WiOpt, 2016.

\bibitem{ShavivNguyenOzgur2015}
D.~Shaviv, P.-M. Nguyen, and A.~{\"{O}}zg{\"{u}}r, ``{Capacity of the Energy
  Harvesting Channel with a Finite Battery},'' arXiv:1506.02024 [cs.IT], Jun.
  2015.

\end{thebibliography}

\appendices

\section{Upper Bound on the Optimal Throughput}
\label{sec:upper_bound}

\begin{proof}[Proof of Proposition~\ref{prop:upperbound}]
Recall $\tilde{E}_t=\min\{E_t,\bar{B}\}$ and $\mu=\mathbb{E}[\tilde{E}_t]$.
For any $n$ and any policy $\mathbf{g}$, we have:
\begin{align*}
\T_n(\mathbf{g})
&=
	\frac{1}{n}\sum_{t=1}^n \mathbb{E}
	\left[\frac{1}{2}\log
	\big(1+\gamma g_t(\tilde{E}^n)\big)\right]\\
&\overset{\text{(i)}}{\leq}
	\frac{1}{2}\log\left(1+\frac{\gamma}{n}
	\mathbb{E}\big[\sum_{t=1}^{n} 
	g_t(\tilde{E}^n)\big]\right)\\
&\overset{\text{(ii)}}{\leq}
	\frac{1}{2}\log\left(1+ \frac{\gamma}{n}\mathbb{E}
	\big[\bar{B}+\sum_{t=2}^{n} \tilde{E}_t\big]
	\right)\\
&=\frac{1}{2}\log\left(1+ \frac{\gamma}{n}\bar{B}+\gamma\tfrac{n-1}{n}\mu\right)
\end{align*}
where (i) is by concavity of $\log$; (ii) follows from the fact that the total allocated energy up to time $n$, can not exceed the total energy that arrives up to time $n$ plus the energy initially available in the battery, which cannot be more than $\bar{B}$:
\[
\sum_{t=1}^{n} g_t\leq \bar{B}+\sum_{t=2}^{n} \tilde{E}_t.
\]
The last term tends to $\frac{1}{2}\log(1+\gamma\mu)$ as $n\to\infty$. Note that this is true for any energy arrival process $E_t$. We therefore have:
\begin{equation*}
\Conline\leq\frac{1}{2}\log(1+\gamma\mu),
\end{equation*}
where $\mu=\mathbb{E}[\min\{E_t,\bar{B}\}]$.
\end{proof}

\section{Throughput Does Not Depend on Initial Battery State}
\label{sec:initial_battery_state}

We state and prove the following proposition:
\begin{proposition}
Let $\mathbf{g}$ be a policy which achieves throughput
$
\liminf_{n\to\infty}\T_n(\mathbf{g})
$
when the initial battery level is $b_1=\bar{B}$.
Then for every $\epsilon>0$ there exists a policy $\tilde{\mathbf{g}}$ which achieves throughput
\[
\liminf_{n\to\infty}\T_n(\tilde{\mathbf{g}})
\geq \liminf_{n\to\infty}\T_n(\mathbf{g})-\epsilon,
\]
when the initial battery level is $b_1=0$.
\end{proposition}
Since a policy which is admissible for $b_1=0$ is also admissible for any $b_1\in[0,\bar{B}]$,
this implies that we can compute the throughput $\Conline$ for any initial battery level, say $\bar{B}$, regardless of the actual battery level of interest $b_1$.
\begin{proof}
Fix $\ell\geq1$. Consider the following online power control policy $\tilde{\mathbf{g}}$ for initial battery level $b_1=0$: Transmit zeros ($g_t=0$) for the first $\ell$ time slots. This will allow the battery to completely recharge to $\bar{B}$ with high probability. Then, if $b_{\ell}=\bar{B}$, transmit the policy $\mathbf{g}$ (note that $b_\ell=\bar{B}$ and $g_\ell=0$ imply $b_{\ell+1}=\bar{B}$). Otherwise, transmit zeros (i.e. give up on the entire transmission).
More precisely, define the new policy as follows, for $t=1,2,\ldots$:
\[
\tilde{g}_t(E^t)=\begin{cases}
0&,1\leq t\leq \ell\\
1_{\{b_{\ell}=\bar{B}\}}\cdot {g}_{t-\ell}(E_{\ell+1}^{t})
&,\ell+1\leq t
\end{cases}
\]
where $1_{\{\cdot\}}$ is the indicator function.
Observe that $b_{\ell}$ is a deterministic function of $E^{\ell}$, which is given by
$b_{\ell}=\min\big\{\sum_{t=2}^{\ell}E_t,\bar{B}\big\}$.
We have for any $n>\ell$:
\begin{align}
\T_n(\tilde{\mathbf{g}})
&=\frac{1}{n}\sum_{t=1}^{n}\mathbb{E}\left[\frac{1}{2}\log\big(1+\tilde{g}_t(E^t)\big)\right]\nonumber\\
&=\frac{1}{n}\sum_{t=\ell+1}^{n}\mathbb{E}\left[\frac{1}{2}\log\big(1+1_{\{b_{\ell}=\bar{B}\}}\cdot{g}_{t-\ell}(E_{\ell+1}^t)\big)\right]\nonumber\\
&=\frac{1}{n}\sum_{t=1}^{n-\ell}\mathbb{E}\left[1_{\{b_{\ell}=\bar{B}\}}\cdot\frac{1}{2}\log\big(1+{g}_t(E_{\ell+1}^{\ell+t})\big)\right]\nonumber\\
&\overset{\text{(i)}}{=}\frac{1}{n}\sum_{t=1}^{n-\ell}\Pr\{b_{\ell}=\bar{B}\}\cdot\mathbb{E}\left[\frac{1}{2}\log(1+{g}_t(E_{\ell+1}^{\ell+t})\big)\right]\nonumber\\
&\overset{\text{(ii)}}{=}\Pr\{b_\ell=\bar{B}\}\cdot\frac{1}{n}\sum_{t=1}^{n-\ell}\mathbb{E}\left[\frac{1}{2}\log(1+{g}_t(E^{t})\big)\right]\nonumber\\
&=\Pr\{b_\ell=\bar{B}\}\cdot\frac{n-\ell}{n}\T_{n-\ell}(\mathbf{g}),
\label{eq:T0_geq_plTBbar}
\end{align}
where (i) is because $b_\ell$ depends only on $E^\ell$, and $E_t$ is i.i.d.;
and (ii) is because $E_t$ is i.i.d.

Since $\mathbb{E}[E_t]>0$, we can lower-bound the probability of recharging the battery using the law of large numbers:
\begin{align*}
\Pr\{b_\ell=\bar{B}\}
&=1-\Pr\big\{\sum_{t=2}^{\ell}E_t<\bar{B}\big\}\\*
&\geq 1 - \epsilon_\ell,
\end{align*}
where $\epsilon_\ell\to0$ as $\ell\to\infty$.
Substituting this in~\eqref{eq:T0_geq_plTBbar} and taking $n\to\infty$ yields
\[
\liminf_{n\to\infty}\T_n(\tilde{\mathbf{g}})
\geq(1-\epsilon_\ell)\liminf_{n\to\infty}\T_n(\mathbf{g}).
\]
By choosing $\ell$ large enough, we can approach the throughput of $\mathbf{g}$ arbitrarily close for any initial battery level.
\end{proof}

\section{Optimal Throughput for Bernoulli Energy Arrivals: Proof of Theorem~\ref{thm:Bernoulli}}
\label{sec:kkt_solution}

Recall Proposition~\ref{prop:Bellman} and the preceding discussion.
It can be argued~\cite[Theorem 6.4]{arapostathis1993discrete} that the there exists an stationary policy, i.e. there exists a function $g^\star(b)$, satisfying $0\leq g^\star(b)\leq b$ for $0\leq b\leq\bar{B}$, s.t. the optimal throughput is given by
\[
\Conline=\liminf_{n\to\infty}\frac{1}{n}\sum_{t=1}^{n}\mathbb{E}[\tfrac{1}{2}
	\log(1+\gamma g^\star(b_t))].
\]
Under such a stationary policy, the battery state $b_t$ is a regenerative process (see Definition~\ref{def:regenerative}). The regeneration times $\{T(n)\}_{n=0}^{\infty}$ are the energy arrivals, i.e. $E_{T(n)}=\bar{B}$, which implies $b_{T(n)}=\bar{B}$.
Additionally, we assume $b_1=\bar{B}$ (See Appendix~\ref{sec:initial_battery_state}),
which implies the process is non-delayed (i.e. the first regeneration time is $T(0)=1$).
Applying Lemma~\ref{lemma:SLLN_regenerative} in Section~\ref{subsec:Bernoulli}, we obtain:
\begin{align*}
\Conline&=\frac{1}{\mathbb{E}L}\mathbb{E}\left[\sum_{t=1}^{L}
	\tfrac{1}{2}\log(1+\gamma g^\star(b_t))\right],
\end{align*}
where $L=T(1)-T(0)$ is a $\text{Geometric}(p)$ RV, which follows from the fact that $E_t$ are i.i.d. $\text{Bernoulli}(p)$. 
Observe that for $2\leq t\leq L$ there are no energy arrivals. Hence, we have the following deterministic recursive relation:
\begin{equation}
\begin{aligned}
b_1&=\bar{B},&&\\
b_t&=b_{t-1}-g^\star(b_{t-1})&&,t=2,\ldots,L.
\end{aligned}
\label{eq:bt_recursive_relation}
\end{equation}
Since $L$ can take any positive integer, this defines a sequence $\{\mathcal{E}^\star_i\}_{i=1}^{\infty}$ such that $g^\star(b_i)=\mathcal{E}^\star_i$.
We can therefore write
\begin{align*}
\Conline&=\frac{1}{\mathbb{E}L}\mathbb{E}\left[
	\sum_{i=1}^{L}\tfrac{1}{2}\log(1+
	\gamma\mathcal{E}^\star_i)\right]\\
&=p\sum_{k=1}^{\infty}p(1-p)^{k-1}\sum_{i=1}^{k}\tfrac{1}{2}
	\log(1+\gamma\mathcal{E}^\star_i)\\
&=\sum_{i=1}^{\infty}p(1-p)^{i-1}\tfrac{1}{2}\log(1+\gamma\mathcal{E}^\star_i).
\end{align*}
Moreover, by the constraint $g^\star(b_t)\leq b_t$ and the recursive relation~\eqref{eq:bt_recursive_relation}, we must have $\sum_{i=1}^{\infty}\mathcal{E}^\star_i\leq\bar{B}$, in addition to ${\mathcal{E}^\star_i\geq0}$ for all $i\geq1$.

To find $\{\mathcal{E}_i^\star\}_{i=1}^{\infty}$ we need to solve the following infinite-dimensional optimization problem:
\begin{equation}
\begin{aligned}
	\text{maximize}
		&\qquad\sum_{i=1}^{\infty}p(1-p)^{i-1}\tfrac{1}{2}\log(1+\gamma\mathcal{E}_i)\\
	\text{subject to}&\qquad
		\mathcal{E}_i\geq0,\quad i=1,2,\ldots,\\
		&\qquad\sum_{i=1}^{\infty}\mathcal{E}_i\leq \bar{B}.
\end{aligned}
\label{eq:def_Cinf}
\end{equation}
Let $\{\mathcal{E}_i^\star\}_{i=1}^{\infty}$ and $\Theta$ be the optimal sequence and optimal objective, respectively, of~\eqref{eq:def_Cinf}.
We will show that~\eqref{eq:def_Cinf} can be solved by the limit as $N\to\infty$ of the following $N$-dimensional optimization problem:
\begin{equation}
\begin{aligned}
	\text{maximize}
		&\qquad\sum_{i=1}^{N}p(1-p)^{i-1}\tfrac{1}{2}\log(1+\gamma\mathcal{E}_i)\\
	\text{subject to}&\qquad
		\mathcal{E}_i\geq0,\quad i=1,2,\ldots,N,\\
		&\qquad\sum_{i=1}^{N}\mathcal{E}_i\leq \bar{B}.
\end{aligned}
\label{eq:def_Cbar}
\end{equation}
Denote by $\Theta_N$ the optimal objective of~\eqref{eq:def_Cbar}.
Clearly $\Theta_N$ is non-decreasing and $\Theta_N\leq\Conline$.
Observe that the first $N$ values of the infinite-dimensional solution, $\{\mathcal{E}_i^\star\}_{i=1}^{N}$, are a feasible solution for~\eqref{eq:def_Cbar}. Therefore,
\begin{align*}
\Theta_N&\geq \sum_{i=1}^{N}p(1-p)^{i-1}\tfrac{1}{2}\log(1+\gamma\mathcal{E}_i^\star)\\
&=\Conline-\sum_{i=N+1}^{\infty}p(1-p)^{i-1}\tfrac{1}{2}
	\log(1+\gamma\mathcal{E}_i^\star)\\
&\overset{(\ast)}{\geq}\Conline
	-\sum_{i=N+1}^{\infty}p(1-p)^{i-1}\tfrac{1}{2}
	\log(1+\gamma\bar{B})\\
&=\Conline-(1-p)^{N}\tfrac{1}{2}\log(1+\gamma\bar{B}),
\end{align*}
where $(\ast)$ is because $\mathcal{E}_t^\star\leq\bar{B}$.
Along with the inequality $\Theta_N\leq\Theta$, this implies
\[
\Theta=\lim_{N\to\infty}\Theta_N.
\]

We continue with the explicit solution of~\eqref{eq:def_Cbar}.
Writing the problem in standard form and using KKT conditions, we have for $i=1,\ldots,N$:
\begin{equation}
	p(1-p)^{i-1}\frac{1}{2}\frac{\gamma}{1+\gamma\mathcal{E}_i}+\lambda_i-\tilde{\lambda}=0,
\label{eq:KKT_equation}
\end{equation}
with $\lambda_i,\tilde{\lambda}\geq0$ and the complementary slackness conditions:
$\lambda_i\mathcal{E}_i=0$ and $\tilde{\lambda}(\sum_{i=1}^{N}\mathcal{E}_i-\bar{B})=0$.

To obtain the non-zero values of $\mathcal{E}_i$, we set $\lambda_i=0$ in~\eqref{eq:KKT_equation}:
\begin{equation}\label{eq:Esolution}
	\mathcal{E}_i=\frac{p(1-p)^{i-1}}{2\tilde{\lambda}}-
	\frac{1}{\gamma}.
\end{equation}
Since $\mathcal{E}_i\geq0$, this implies $\tilde{\lambda}\leq\frac{\gamma}{2}p(1-p)^{i-1}$ for all $i$ for which $\mathcal{E}_i>0$.
Since the RHS is a decreasing function of $i$, there exists an integer $\tilde{N}$ such that $\mathcal{E}_i>0$ for if $i\leq\tilde{N}$ and $\mathcal{E}_i=0$ otherwise.
Therefore $\tilde{N}$ is the largest integer satisfying $\tilde{N}\leq N$ and
\begin{equation}
\tilde{\lambda}\leq\frac{\gamma}{2}p(1-p)^{\tilde{N}-1}.
\label{eq:Ntilde_bound}
\end{equation}


Next, consider the total energy constraint $\sum_{i=1}^{N}\mathcal{E}_i\leq\bar{B}$.
Since increasing $\mathcal{E}_i$ for any $i$ will only increase the objective, this constraint must hold with equality:
\begin{align*}
	\bar{B}&=\sum_{i=1}^{N}\mathcal{E}_i\\
	&=\sum_{i=1}^{\tilde{N}}
		\left(\frac{p(1-p)^{i-1}}
		{2\tilde{\lambda}}-\frac{1}{\gamma}\right)\\
	&=\frac{1-(1-p)^{\tilde{N}}}{2\tilde{\lambda}}
		-\frac{\tilde{N}}{\gamma}.
\end{align*}
This yields:
\begin{equation}\label{eq:lambda_tilde}
	\tilde{\lambda}=\frac{1-(1-p)^{\tilde{N}}}
	{2(\bar{B}+\tilde{N}/\gamma)}.
\end{equation}

Along with~\eqref{eq:Ntilde_bound}, we deduce that $\tilde{N}$ is the largest integer satisfying $\tilde{N}\leq N$ and
\[
	\frac{1-(1-p)^{\tilde{N}}}{2(\bar{B}+\tilde{N}/\gamma)}
	\leq \frac{\gamma}{2}p(1-p)^{\tilde{N}},
\]
or equivalently
\[
	1 \leq (1-p)^{\tilde{N}}[1+p(\gamma\bar{B}+\tilde{N})].
\]
Observe that for $N$ large enough, the optimal $\tilde{N}$ does not depend on $N$, and therefore $\Theta_N$ will not depend on $N$.
For such values of $N$, the optimal sequence $\{\mathcal{E}_i\}_{i=1}^{N}$ for~\eqref{eq:def_Cbar} is given by substituting~\eqref{eq:lambda_tilde} in~\eqref{eq:Esolution} for $i=1,\ldots,\tilde{N}$, and $\mathcal{E}_i=0$ for $i>\tilde{N}$. 
Since $\Theta_N$ does not depend on $N$ and $\Theta=\lim_{N\to\infty}\Theta_N$,
the optimal sequence for~\eqref{eq:def_Cinf} is the same.
This gives~\eqref{eq:Bernoulli_gt_solution}.

\end{document}